\documentclass{llncs}

\newcommand{\maxdeg}{{\ensuremath{\Delta}}\xspace}
\newcommand{\mopt}{\ensuremath{{M^*}}\xspace}
\newcommand{\m}{\ensuremath{M}\xspace}
\newcommand{\oneOverTwo}{\ensuremath{{\frac{1}{2}}}\xspace}

\newcommand{\twoOverThree}{\ensuremath{{\frac{2}{3}}}\xspace}
\newcommand{\desiredRatio}{\ensuremath{{\frac{\maxdeg-1}{2\maxdeg-3}}}\xspace}
\newcommand{\weakerRatio}{\ensuremath{{\frac{\maxdeg-1/2}{2\maxdeg-2}}}\xspace}
\newcommand{\transferred}{\ensuremath{\theta}\xspace}
\newcommand{\mingreedy}{\textsc{MinGreedy}\xspace}
\newcommand{\karpsipser}{\textsc{KarpSipser}\xspace}
\newcommand{\mrg}{\textsc{MRG}\xspace}

\newcommand{\greedy}{\textsc{Greedy}\xspace}
\newcommand{\ranking}{\textsc{Ranking}\xspace}
\newcommand{\shuffle}{\textsc{Shuffle}\xspace}
\newcommand{\greedymatching}{\ensuremath{\mathcal{APV}}\xspace}
\newcommand{\g}{{\ensuremath{G}}\xspace}
\newcommand{\mg}{{\ensuremath{H}}\xspace}
\newcommand{\w}[1]{{\ensuremath{m_{#1}}}\xspace}
\newcommand{\wopt}[1]{{\ensuremath{m_{#1}^*}}\xspace}
\newcommand{\mgapprox}{{\ensuremath{\alpha}}\xspace}

\newcommand{\edge}[1]{\ensuremath{\{#1\}}\xspace}
\newcommand{\dedge}[1]{\ensuremath{(#1)}\xspace}

\newcommand{\inlineheading}[1]{\smallskip\textbf{#1}}

\newcommand{\oneonepath}{singleton\xspace}
\newcommand{\oneonepaths}{singletons\xspace}
\newcommand{\onetwopath}{\oneOverTwo-path\xspace}

\newcommand{\mmoptpath}{augmenting path\xspace}
\newcommand{\mmoptpaths}{augmenting paths\xspace}

\newcommand{\onetwomingreedy}{\mbox{\textsc{1-2-MinGreedy}}\xspace}
\newcommand{\degreeone}{\mbox{degree-1}\xspace}
\newcommand{\don}[1]{\ensuremath{\boldsymbol{(#1)}}\xspace}

\usepackage{ntheorem}
\theoremseparator{.}
\newtheorem{observation}{Observation}

\usepackage{amsmath}
\usepackage{alltt}
\usepackage[T1]{fontenc}
\usepackage{xspace}
\usepackage{tikz}
\usetikzlibrary{calc}
\usetikzlibrary{arrows}
\usetikzlibrary{decorations.markings}
\usetikzlibrary{decorations.pathmorphing}
\usetikzlibrary{decorations.pathreplacing}
\usetikzlibrary{positioning}
\usetikzlibrary{shapes}
\usepackage{enumitem}
\usepackage{hyperref}
\usepackage{cleveref}
\usepackage{bm}

\tikzset{
 every node/.style={
  circle,
  draw=black,
  minimum size=4mm,
  inner sep=0mm
 },
 opt/.style={
  double=white,
  double distance=2pt
 },
 mg/.style={
  postaction={
   decorate,
   decoration={
    markings,
    mark=at position 0.5 with {
     \draw (-0.04,0.15) -- (-0.04,-0.15)
           ( 0.04,0.15) -- ( 0.04,-0.15);
    }
   }
  }
 },
 transfer/.style={
  >=triangle 45,
  ->
 },
 nonMcov/.style={
  fill=lightgray
 },
 altpath/.style={
  decoration={snake},
  decorate
 },
 alabel/.style={
  rectangle,
  fill=none,
  draw=none
 }
}

\crefname{section}{Section}{Sections}
\crefname{remark}{Remark}{Remarks}
\crefname{proposition}{Proposition}{Propositions}
\crefname{lemma}{Lemma}{Lemmas}
\crefname{theorem}{Theorem}{Theorems}
\crefname{definition}{Definition}{Definitions}
\crefname{figure}{Figure}{Figures}
\crefname{appendix}{Appendix}{Appendices}
\crefname{table}{Table}{Tables}
\crefname{equation}{\unskip}{\unskip}
\crefname{corollary}{Corollary}{Corollaries}
\crefname{observation}{Observation}{Observations}
\crefname{enumi}{\unskip}{\unskip}

\title{\large Approximation Bounds For Minimum Degree Matching}
\author{Bert Besser\thanks{Partially supported by DFG SCHN 503/6-1.}}
\institute{Institut f\"ur Informatik, Goethe-Universit\"at Frankfurt am Main, Germany}

\begin{document}

\maketitle

\begin{abstract}
\vspace{-.5cm}
We consider the \mingreedy strategy for Maximum Cardinality Matching. \mingreedy repeatedly selects an edge incident with a node of minimum degree. For graphs of degree at most $ \maxdeg $ we show that \mingreedy achieves approximation ratio at least $ \frac{\maxdeg-1}{2\maxdeg-3} $ in the worst case and that this performance is optimal among adaptive priority algorithms in the vertex model, which include many prominent greedy matching heuristics.

Even when considering expected approximation ratios of randomized greedy strategies, no better worst case bounds are known for graphs of small degrees.
\\[.4\baselineskip]
{\bf Keywords:} matching, greedy, approximation, priority algorithm.
\end{abstract}

\section{Introduction}
\label{intro}

In the \inlineheading{Maximum Cardinality Matching} Problem a node disjoint subset of edges of maximum size is to be determined.
Matching problems have many applications, e.g. image feature matching in computer vision or protein structure comparison in computational biology.

A maximum matching can be found in polynomial time, e.g. by the algorithm of Micali and Vazirani \cite{DBLP:conf/focs/MicaliV80} running in time $ O(\sqrt{|V|}\cdot |E|\cdot\alpha(|E|,|V|)) $ where $ \alpha(|E|,|V|) $ is the inverse Ackermann function \cite{DBLP:journals/corr/abs-1210-4594}.
The algorithm of Mucha and Sankowski \cite{DBLP:conf/focs/MuchaS04} runs asymptotically faster on dense graphs, its runtime is $ O(|V|^\omega) $ where $ \omega<2.38 $ is the exponent needed to perform matrix multiplication.

However, there are much faster greedy algorithms that in practice compute very large matchings, even near optimal ones.
Very good approximate solutions may already be satisfactory in some applications.
If maximum matchings are needed, one can save lots of runtime when feeding large greedy matchings into optimal algorithms which iteratively improve an initial solution.

\begin{figure}
\centering
\begin{minipage}{.8\textwidth}
\footnotesize
\begin{alltt}
\(\m=\emptyset\)
repeat until all edges removed from input graph:
   select (random) node \(u\) of minimum non-zero degree
   select (random) neighbor \(v\) of \(u\)
   pick edge \(\{u,v\}\), i.e.\ set \(\m=\m\cup\{\{u,v\}\}\)
   remove all edges incident with \(u\) and \(v\) from input graph
return \(\m\)
\end{alltt}
\end{minipage}
\caption{The (randomized) \mingreedy algorithm}
\label{algo:mingreedy}
\end{figure}

\inlineheading{\mingreedy.}
The (randomized) \mingreedy algorithm computes a matching \m by repeatedly picking edges incident to nodes of currently minimum degree, see \cref{algo:mingreedy}.
\mingreedy can be implemented in linear time $ O(|V|+|E|) $.

The experimentally observed approximation performance is quite impressive.
Tinhofer \cite{tinhofer} observed that on random graphs of varying density \mingreedy performed superior to \greedy (which randomly selects an edge) and to~\mrg (which randomly selects a node and subsequently an incident edge).
In experiments of Frieze et al. \cite{Frieze:1993:ASG:313559.313800} on random cubic graphs \mingreedy left only about 10 out of $ 10^6 $ nodes unmatched.
On random graphs of small constant average degree Magun \cite{Magun97greedymatching} observed that \mingreedy produces extremely few ``lost edges'' in comparison with an optimal solution.

In an involved argument, Frieze et al. \cite{Frieze:1993:ASG:313559.313800} showed that $c_1\cdot  n^{\frac{1}{5}}\leq\lambda_n\leq c_2\cdot n^{\frac{1}{5}}\cdot\ln(n)$ holds ($c_1,c_2$ being constants) for the expected number $\lambda_n$ of nodes not being matched by \mingreedy on random $ n $-node cubic graphs.

Whereas the performance of \mingreedy on random instances is very good, its worst case performance is poor.
Poloczek \cite{matthiasDiss} constructed hard input instances on which \mingreedy (and common variations of the algorithm) achieves approximation ratio at most $ \oneOverTwo+o(1) $ w.h.p.

\inlineheading{\mrg, \greedy \& \shuffle.}
The Modified Random Greedy algorithm, abbreviated~\mrg in literature, ignores node degrees and repeatedly selects a node and then a neighbor uniformly at random.
The expected approximation ratio was shown to be at least $\oneOverTwo+\frac{1}{400.000}$ by Aronson et al. \cite{RSA:RSA3240060107}. 

The random edge algorithm \greedy repeatedly selects an edge uniformly at random.
For graphs with degrees bounded by \maxdeg an expected lower bound on the approximation ratio of $\frac{\maxdeg}{2\maxdeg-1}$ was shown by Dyer and Frieze \cite{DBLP:journals/rsa/DyerF91} and later improved by Miller and Pritikin \cite{DBLP:journals/rsa/MillerP97} to $ \oneOverTwo(\sqrt{(\maxdeg-1)^2+1}-\maxdeg+2) $.
If \greedy prefers edges of degree-1 nodes, the \karpsipser algorithm is obtained, which is asymptotically optimal w.h.p. on large sparse random graphs \cite{DBLP:conf/focs/KarpS81}.

The \shuffle algorithm, proposed by Goel and Tripathi \cite{DBLP:conf/focs/GoelT12}, is an adaptation of the \ranking algorithm of Karp et al. \cite{DBLP:conf/stoc/KarpVV90} to non-bipartite graphs.
\shuffle selects a random permutation $ \pi $ of the nodes and repeatedly matches the, according to $ \pi $, first non-isolated node to its first unmatched neighbor.
Chan et al. \cite{DBLP:conf/soda/ChanCWZ14} showed that \shuffle achieves an approximation ratio of at least $ 2\cdot(5-\sqrt{7})/9\approx0.523 $.

\inlineheading{Inapproximability.}
To show performance bounds for greedy algorithms, Borodin et al. \cite{DBLP:conf/soda/BorodinNR02} proposed the model of \emph{adaptive priority algorithms}.
The model formalizes the greedy nature of an algorithm:
while gathering knowledge about the input, irrevocable decisions have to be made to construct a solution.

Davis and Impagliazzo \cite{DBLP:conf/soda/DavisI04} introduced the \emph{vertex model} to study adaptive priority algorithms for graph problems.
Adaptive priority algorithms in the vertex model for the matching problem, which we call \greedymatching-algorithms, implement powerful node and edge selection routines. 
In particular, in each step a node $ v $ and an incident edge is not picked arbitrarily but based on all knowledge already gathered about $ v $ and its neighbors, e.g. is a neighbor matched or unmatched, what is the degree of a neighbor, what are the neighbors of a neighbor, etc.
\greedymatching-algorithms include~\greedy, \karpsipser, \mrg, \mingreedy, \shuffle and all vertex iterative algorithms, a class of algorithms defined in \cite{DBLP:conf/focs/GoelT12} as a generalization of \shuffle.

Despite the strength of \greedymatching-algorithms, Poloczek \cite{matthiasDiss} constructed rather simple graphs with worst case approximation ratio at most \twoOverThree.
(He also showed an inapproximability bound of $ \frac{5}{6} $ for randomized priority algorithms.)
Presented in the same thesis, for graphs with arbitrarily large degree Besser and Poloczek showed that no \greedymatching-algorithm achieves worst case approximation ratio better than $\frac{1}{2} + o(1)$.

\inlineheading{Contributions.}
From now on we reserve the name \mingreedy for the deterministic version of \mingreedy in which a node of minimum degree and an incident edge is picked by a worst case adversary.

\smallskip

We show that \mingreedy approximates an optimal matching within a factor of \desiredRatio for graphs in which degrees are bounded by at most \maxdeg.
In the proof we analyze a variant of \mingreedy which is also related with the \karpsipser algorithm.

\smallskip

We also show that the worst case approximation performance of \mingreedy is optimal for (deterministic) \greedymatching-algorithms.
In particular, we improve the construction of Besser and Poloczek given in \cite{matthiasDiss} and present hard input instances of degree at most \maxdeg for which any \greedymatching-algorithm computes a matching of size at most $ \desiredRatio+o(1) $ times optimal.

\smallskip

Our worst case performance guarantees are stronger than the best known bounds on the expected performance of \mrg and \shuffle ($ \oneOverTwo{+}\frac{1}{400.000} $ resp.\ $ {\approx}0.523 $), for small \maxdeg, and of \greedy ($ {\approx}0.618 $ for $ \maxdeg{=}3 $), for all~\maxdeg.

\inlineheading{Techniques.}
For our performance guarantees for \mingreedy we study the matching graph composed of the edges of a matching \m, computed by \mingreedy, and of a maximum matching \mopt.
The connected components are alternating paths and cycles.
Only paths of length three have poor ``local'' approximation ratio (of \m-edges to \mopt-edges).
To obtain a global performance guarantee, we balance local approximation ratios by transferring ``\m-funds'' from rich to poor components using edges of the input graph.

Incorporating the properties of \mingreedy within an amortized analysis is our technical contribution.

\inlineheading{Overview.}
In \cref{sect:theChargingScheme} we present the charging scheme used to prove the performance guarantees for \mingreedy.
In \cref{tightPerfBounds} we show our \twoOverThree bound for graphs of degree at most~$ \maxdeg=3 $.
For graphs of degree at most $ \maxdeg\geq4 $ we present in \cref{sect:onetwomingreedy4} our performance guarantee of \desiredRatio.
Our inapproximability results for \greedymatching-algorithms are given in \cref{hardinstances}.

\pagebreak
\section{The Charging Scheme}
\label{sect:theChargingScheme}

\inlineheading{The Matching Graph.}
Let $\g=(V,E)$ be a connected graph, \mopt a maximum matching in \g and \m a matching computed by \mingreedy when applied to input \g.
To analyze the worst case approximation ratio of \mingreedy we investigate the graph
$$ \mg=(V,\m\cup\mopt)\,. $$
The connected components of \mg are paths and cycles composed of edges of \m and \mopt.
For example, \mg contains so-called \emph{(\m-)\mmoptpaths}, or \emph{paths} for short:
an \mmoptpath $ X $ has $ \w{X} $ edges of \m and $\wopt{X}=\w{X}+1$ edges of \mopt and starts and ends with an \mopt-edge:
\begin{center}

\reflectbox{\scalebox{.75}{
\begin{tikzpicture}
\node (mopt1) {};
\node[right of=mopt1] (mopt2) {};
\node[right of=mopt2,right of=mopt2] (mopt3) {};
\node[right of=mopt3] (mopt4) {};
\node[right of=mopt4] (mopt5) {};
\node[right of=mopt5] (mopt6) {};
\node[right of=mopt6] (mopt7) {};

\draw
(mopt1) edge[opt] (mopt2)
(mopt2) edge[altpath] (mopt3)
(mopt3) edge[mg] (mopt4)
(mopt4) edge[opt] (mopt5)
(mopt5) edge[mg] (mopt6)
(mopt6) edge[opt] (mopt7)
;
\end{tikzpicture}
}}

\end{center}
We call the two nodes of a path~$ X $ which do not have an incident~\m-edge the \emph{endpoints} of~$ X $ (the leftmost and the rightmost node in the figure).
Other connected components of \mg are edges of $ \m\cap\mopt $
\begin{center}
\reflectbox{\scalebox{.75}{
\begin{tikzpicture}

\node (s1) {};
\node[right of=s1] (s2) {};

\draw (s1) edge[opt,mg] (s2);
\end{tikzpicture}
}},
\end{center}
which we call \emph{\oneonepaths}.
For a \oneonepath $ X $ we have $ \w{X}=\wopt{X}=1 $.
We may focus on these two component types:
\begin{lemma}
There is a maximum matching \mopt in \g such that each connected component of \mg is a \oneonepath or an \mmoptpath.
\end{lemma}
\begin{proof}
Let $ M' $ be a maximum matching in $\g=(V,E)$ and let \mingreedy compute the matching \m.
We show how to transform $ M' $ into \mopt.
The connected components of the graph $ (V,\m\cup M') $ are \oneonepaths and alternating paths and cycles, where a path starts and ends with an \m-edge or $ M' $-edge and a cycle does not have path endpoints.
To prove the statement we eliminate paths starting and ending both with an \m-edge, paths starting and ending with different types of edges, and cycles.
There is no path $ X $ of the first type, since if there was, then we could replace the $ m_{X}' $ many $ M' $-edges of $ X $ with the $ m_X'+1 $ many \m-edges of $ X $ to obtain a matching larger than $ M' $.
In a component $ X $ of the latter types we replace the $ M' $-edges of $ X $ with the \m-edges of $ X $:
component $ X $ is replaced by $ m_X $ many \oneonepaths.
\qed
\end{proof}

\inlineheading{Local Approximation Ratios.}
To bound the approximation ratio
$$ \mgapprox=|\m|/|\mopt| $$
of \mingreedy, our approach is to bound \emph{local approximation ratios}
$$ \mgapprox_X=\w{X}/\wopt{X}\,, $$
of components $ X $ of \mg.
Observe that a \emph{\onetwopath} $ X $
\begin{center} \reflectbox{\scalebox{.75}{
\begin{tikzpicture}

\node (bad1) {};
\node[right of=bad1] (bad2) {};
\node[right of=bad2] (bad3) {};
\node[right of=bad3] (bad4) {};

\draw
(bad1) edge[opt] (bad2)
(bad2) edge[mg] (bad3)
(bad3) edge[opt] (bad4)
;

\end{tikzpicture}
}},
\end{center}i.e.\ an \mmoptpath with $ \w{X}=1 $ edge of \m and $ \wopt{X}=2 $ edges of \mopt, has local approximation ratio $ \mgapprox_X=\oneOverTwo $, while all other components have local approximation ratios at least \twoOverThree.
In particular, a \oneonepath $ X $ has optimal local approximation ratio $ \mgapprox_X=1 $.
Thus we have to balance local approximation ratios.
For any component $ X $, we say that $ X $ \emph{has \m-funds} \w{X} and we introduce a change $ t_X $ to the \m-funds of $ X $ such that the changed local approximation ratio of $ X $ is lower bounded by
$$ \mgapprox_X=\frac{\w{X}+t_X}{\wopt{X}}\stackrel{!}{\geq}\beta$$
for appropriately chosen $ \beta $.
If $ \sum_Xt_X=0 $ holds, then the total \m-funds $ \sum_X \w{X}+t_X=\sum_X\w{X}=|\m|$ are unchanged and hence \mingreedy achieves approximation ratio at least
$$ \mgapprox=|\m|/|\mopt|=\left(\sum_X\w{X}+t_X\right)/|\mopt|\geq\left(\sum_X\beta\wopt{X}\right)/|\mopt|=\beta\,. $$

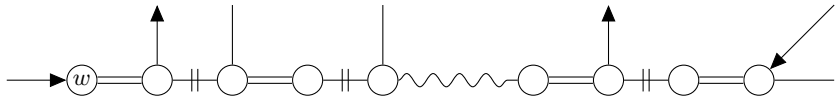
\begin{figure}[htbp!]
\centering
\scalebox{1}{
\begin{tikzpicture}
\node              (v1) {$ w $};
\node[right of=v1] (v2) {};
\node[right of=v2] (v3) {};
\node[right of=v3] (v4) {};
\node[right of=v4] (v5) {};
\node[right of=v5,right of=v5] (v6) {};
\node[right of=v6] (v7) {};
\node[right of=v7] (v8) {};
\node[right of=v8] (v9) {};

\draw[opt] (v1) -- (v2) (v3) -- (v4) (v6) -- (v7) (v8) -- (v9);
\draw[mg] (v2) -- (v3);
\draw[mg] (v4) -- (v5);
\draw[mg] (v7) -- (v8);

\draw[transfer,->] ($(v1)+(-1,0)$) -- (v1);
\draw[transfer,->] ($(v9)+(1,1)$) -- (v9);
\draw[] ($(v9)+(1,0)$) -- (v9);
\draw[transfer,->] (v2) -- ($(v2)+(0,1)$);
\draw[] (v3) -- ($(v3)+(0,1)$);
\draw[] (v5) -- ($(v5)+(0,1)$);
\draw[transfer,->] (v7) -- ($(v7)+(0,1)$);

\draw[altpath] (v5) -- (v6);
\end{tikzpicture}
}
\caption{
An \mmoptpath with transfers to its path endpoints and from its \m-covered nodes.
Not all $ F $-edges are used to move \m-funds.
}
\label{exampleCreditsDebits}
\end{figure}

\inlineheading{Transferring ${\m}$-Funds.}
The idea is to transfer \m-funds from rich components to poor components:
as an example, we could transfer \m-funds from a \oneonepath to a \onetwopath.
But which components should be involved in a transfer?

Observe that an augmenting path endpoint $ w $ is detrimental to its \mmoptpath $ X $, since $ w $ decreases the local approximation ratio of $ X $.
So a transfer should push \m-funds towards poor $ w $ from a rich \m-covered node.

Since the \g-neighbors of $ w $ are \m-covered (otherwise \m would not be maximal), our approach is to move \m-funds to $ w $ from \g-neighbors of $ w $ which belong to other components:
\m-funds are moved over certain edges in
$$ F:=E\setminus(\m\cup\mopt) $$
which connect the components of \mg, see \cref{exampleCreditsDebits} for an illustration.
We verify that a poor component, in particular a \onetwopath, is able to receive \m-funds:
\pagebreak
\begin{lemma}
\label{obs:potentialTransfer}
An \mmoptpath endpoint $ w $ is incident with $ F $-edge, since the degree of $ w $ in \g is at least
$$ d_\g(w)\geq2\,. $$
\end{lemma}
\begin{proof}
When \mingreedy picks the first \m-edge of the \mmoptpath of $ w $, a node of degree at least two is selected (with an incident $ M $-edge and $ M^* $-edge).
So $ w $ also has degree at least two and is incident not only with its \mopt-edge.
\qed
\end{proof}

When transferring $ M $-funds over all $ F $-edges we face the danger of \mmoptpath endpoints having large degree and drawing very large amounts of \m-funds:
rich components might become poor and now have too small local approximation ratio.
The following definition prevents high degree path endpoints from wasting \m-funds.
\begin{definition}
\label{def:transfer}
Let $ w $ be an \mmoptpath endpoint. 
Edge $ \{v,w\}\in F$ is a \emph{transfer} if in the step of \mingreedy matching $ v $ the degree of $ w $ drops to at most $ d(w)\leq1 $.
A transfer moves one \emph{coin} with value~\transferred from the component of~$ v $ to the component of~$ w $.
\end{definition}

We frequently denote a transfer $ \{v,w\} $ as $ (v,w) $ to stress its direction.
In order to refer to transfers from/to a given component $ X $, we also call $ (v,w) $ a \emph{debit from} $ v $ and a \emph{credit to} $ w $.
I.e., debits from a component are transfers directed from its \m-covered nodes, credits to an \mmoptpath are transfers directed to its path endpoints.

\inlineheading{Bounding Local Approximation Ratios.}
Let $ X $ be a component of \mg and denote the numbers of debits and credits of nodes of $ X $ by $ d_X $ respectively $ c_X $.
We call $ c_X-d_X $ the \emph{balance} of $ X $.\footnote{Previous versions of this work use the notion of \emph{total debits} of a component~$ X $, which are defined as~$ d_X{-}c_X $.
This notion is equivalent to the notion of the balance of~$ X $.}
Given an upper bound on the debits from $ X $ and a lower bound on the credits to $ X $, we obtain a bound on the balance of $ X $ of at least
$$ c_X-d_X\geq T_X\,. $$
Whenever we move \m-funds over an edge $ \{v,w\}\in F $, we transfer an amount $\transferred$ of \m-funds to the \mmoptpath endpoint $ w $ using one coin.
Hence the local approximation ratio of $ X $ is at least
$$\mgapprox_X=
\frac{\w{X}+\transferred c_X-\transferred d_X}{\wopt{X}}
\geq\frac{\w{X}+\theta T_X}{\wopt{X}}\,.$$
In the analysis we find $ T_X $ and \transferred such that $ \mgapprox_X\geq\beta $ for all components.
Hence \mingreedy computes matchings of size at least $ \beta $ times optimal.

\pagebreak
\subsection{Balance Bounds}

For each path~$ X $ we demand the following:
both nodes of each~\m-edge~\edge{x,x'} of~$ X $ pay at most
\begin{align}
\label{def:edgedeb}
d_{\edge{x,x'}}\leq2(\maxdeg-2)\tag{$ - $Edge}
\end{align}
coins and~$ X $ receives at least
\begin{align}
\label{eqn:mintwocreds}
c_X\geq2\tag{$ + $Endpoints}
\end{align}
coins.
Therefore, the minimum balance of path~$ X $ is~$ 2-D_X $ for~$ D_X\!:=2\w{X}(\maxdeg-2) $.
In particular, bounds \cref{def:edgedeb,eqn:mintwocreds} 
are bounds on numbers of coins and not on numbers of transfers:
we demand that both hold w.r.t.\ transfers as well as \emph{donations}, the latter of which are another way to move coins that we introduce later.

To show that all components have local approximation ratio at least~$ \frac{\maxdeg-1}{2\maxdeg-3} $ we will verify balance bounds
\begin{align}
-d_X&~\geq-2(\maxdeg-1)~~+~~2&\text{for any singleton~$ X $ and}\tag{$ \pm $Singleton}\label{eqn:balanceboundsinglereg}\\
c_X-d_X&~\geq~ 2-D_X~~~~~\>\!+~~2(\maxdeg-2)&\text{for any path~$ X $.}\tag{$ \pm $Path}\label{eqn:balanceboundaugreg}
\end{align}

\begin{lemma}
If balance bounds \cref{eqn:balanceboundsinglereg,eqn:balanceboundaugreg} hold then all components have local approximation ratio at least~$ \frac{\maxdeg-1}{2\maxdeg-3} $.
\end{lemma}
\begin{proof}
Choose~$ \transferred=\frac{1}{2(2\maxdeg-3)} $.
Using \cref{eqn:balanceboundsinglereg}, the local approximation ratio of a singleton~$ X $ is at least~$ \frac{1+\transferred(c_X-d_X)}{1}\geq1-2\transferred(\maxdeg-2) $, which in turn equals
\begin{align}
\label{eqn:nonbipApproxProof}
1{-}2\transferred(\maxdeg-2)~=~1{-}\frac{2(\maxdeg-2)}{2(2\maxdeg-3)}~=~\frac{\maxdeg-1}{2\maxdeg-3}\,.
\end{align}
Therefore each singleton has local approximation ratio at least~$ \frac{\maxdeg-1}{2\maxdeg-3} $.

Using \cref{eqn:balanceboundaugreg}, we get that the local approximation ratio of a path~$ X $ is at least
\begin{align*}
\frac{\w{X}+\transferred(c_X-d_X)}{\wopt{X}}&\geq\frac{\w{X}-2\transferred\w{X}(\maxdeg-2)+2\kappa(\maxdeg-1)}{\w{X}+1}\\
&=1-2\transferred(\maxdeg-2)+\frac{2\transferred(2\maxdeg-3)-1}{\w{X}+1}\\
&=1-2\transferred(\maxdeg-2)\,,
\end{align*}
which by \cref{eqn:nonbipApproxProof} is bounded by at least~$ \frac{\maxdeg-1}{2\maxdeg-3} $ as well, i.e.\ all paths have local approximation ratio at least~$ \frac{\maxdeg-1}{2\maxdeg-3} $.
\qed
\end{proof}

\subsection{Organisation of the Proof of \cref{eqn:balanceboundsinglereg,eqn:balanceboundaugreg}}

Consider the \onetwomingreedy algorithm given in \cref{algo:onetwomingreedy}.
Observe that a sequence of edges picked by the (deterministic) \mingreedy algorithm is a valid sequence of edges to be picked by \onetwomingreedy.
Therefore \onetwomingreedy achieves no better worst case approximation ratio than \mingreedy and obtaining performance guarantees for \onetwomingreedy implies the same approximation performance for \mingreedy.

Note that \onetwomingreedy can also be understood as a refined variant of the (deterministic) \karpsipser algorithm, which selects an arbitrary edge whenever all degrees are at least two and otherwise selects an arbitrary node of minimum degree, i.e.\ of degree one, and matches it with an arbitrary neighbor.

\begin{figure}[htbp!]
\centering
\begin{minipage}{.8\textwidth}
\footnotesize
\begin{alltt}
\(\m=\emptyset\)
repeat until all edges removed from input graph:
   if each node has degree at least 3:
      select (arbitrary) edge \(\{u,v\}\)
   else:
      select (arbitrary) node \(u\) of minimum non-zero degree
      select (arbitrary) neighbor \(v\) of \(u\)
   pick edge \(\{u,v\}\), i.e.\ set \(\m=\m\cup\{\{u,v\}\}\)
   remove all edges incident with \(u\) and \(v\) from input graph
return \(\m\)
\end{alltt}
\end{minipage}
\caption{The (deterministic) \onetwomingreedy algorithm}
\label{algo:onetwomingreedy}
\end{figure}

\medskip

We prepare the proof with some basic observations in \cref{sect:preparations}.
Then, for \onetwomingreedy we verify bounds~\cref{def:edgedeb}, \cref{eqn:mintwocreds}, \cref{eqn:balanceboundsinglereg}, and \cref{eqn:balanceboundaugreg} for graphs of maximum degree~$ \maxdeg=3 $ in \cref{tightPerfBounds}, where our argument will rely on transfers only.
In \cref{sect:onetwomingreedy4} we analyze graphs of maximum degree~$ \maxdeg\geq4 $, where we will introduce \emph{donations} as an additional way to move coins.

\subsection{Preparations}
\label{sect:preparations}
Since transfers move coins from~\m-covered nodes over~$ F $-edges to adjacent path endpoints, the following bounds on the number of outgoing and incoming transfers are an immediate consequence.

In particular, the following bounds imply \cref{def:edgedeb} as well as \cref{eqn:mintwocreds} in case coins are moved only in transfers but not in donations.
\begin{lemma}
\label{lemma:thebasicboundshold}
Let an~\m-edge~\edge{u,v} be given.\\
If~\edge{u,v} is the edge of a singleton, then both nodes pay at most~$2(\maxdeg-1) $ debits.
If~\edge{u,v} is a path~\m-edge, then both nodes pay at most~$ 2(\maxdeg-2) $ debits.
\end{lemma}
\begin{proof}
Since only~$ F $-edges can be transfers, it suffices to bound the number of~$ F $-edges incident with each of~$ u $ and~$ v $.
Each~\m-covered node of a singleton is incident with at most~$ \maxdeg-1 $ many~$ F $-edges.
Each~\m-covered node of a path is incident with at most~$ \maxdeg-2 $ many~$ F $-edges.
\qed
\end{proof}

\begin{lemma}
\label{prop:numcreditsreg}
Let~$ w $ be an endpoint of path~$ X $.
Endpoint~$ w $ receives at least one credit.
Path~$ X $ receives at least~$ c_X\geq2 $ credits.
\end{lemma}
\begin{proof}
Let~$ d_\g(w) $ denote the degree of~$ w $ in~\g and denote the current degree of~$ w $ by~$ d(w) $.
By Lemma~\ref{obs:potentialTransfer} we have $ d_\g(w)\geq2 $.

Since~$ w $ is a path endpoint, node~$ w $ never gets matched.
Consider the step when the degree of~$ w $ drops from~$ d(w) $ to~0, where~$ d(w)\in\{1,2\} $.

If~$ d(w)=2 $ holds, then an~$ F $-edge incident with~$ w $ gets removed and we obtain~$ c_w\geq1 $.

If~$ d(w)=1 $ holds, then either an~$ F $-edge incident with~$ w $ gets removed and we have~$ c_w\geq1 $, or the~\mopt-edge of~$ w $ gets removed.
In the latter case the degree of~$ w $ dropped from~$ d_\g(w)\geq2 $ to~$ d(w)=1 $ when an~$ F $-edge incident with~$ w $ was removed, and again we have~$ c_w\geq1 $.

Finally, observe that since path~$ X $ has two endpoints, path~$ X $ receives at least~$ c_X\geq2 $ credits.
\qed
\end{proof}

\section{Maximum Degree $ \maxdeg=3 $}
\label{tightPerfBounds}

In this section we show
\begin{theorem}
\label{thm:mingreedyreg}
The \onetwomingreedy algorithm achieves approximation ratio at least~$ \frac{\maxdeg-1}{2\maxdeg-3}=\twoOverThree $ for graphs of maximum degree~$ \maxdeg=3 $.
\end{theorem}

In the proof we verify \cref{eqn:balanceboundsinglereg,eqn:balanceboundaugreg}, which for~$ \maxdeg=3 $ reduce to
\begin{align}
-d_X&\geq-2&\text{for each singleton~$ X $ and}\label{eqn:balanceboundsinglereg3}\tag{\ref{eqn:balanceboundsinglereg}$^3$}\\
c_X-d_X&\geq2-2\w{X}+2&\text{for each path~$ X $.}\label{eqn:balanceboundaugreg3}\tag{\ref{eqn:balanceboundaugreg}$^3$}
\end{align}
Coins will be moved using transfers only.
Therefore \cref{def:edgedeb,eqn:mintwocreds} hold:
in case only transfers are used, we have already verified both bounds in \cref{lemma:thebasicboundshold,prop:numcreditsreg}.

\subsection{Two Debits are Missing}

Recall that the minimum balance of a singleton~$ X $ is~$ -2(\maxdeg-1){=}-4 $, i.e.\ singleton~$ X $ pays at most four debits.
Our plan to verify \cref{eqn:balanceboundsinglereg3} is to show that~$ X $ pays two debits less than maximum:
we say that two debits are \emph{missing}.

Similarly, the minimum balance of a path~$ X $ is~$ 2-D_X=2-2\w{X} $, where~$ D_X $ is the maximum number of debits payed by~$ X $.
Again, if~$ X $ has two missing debits then~$ X $ is balanced as required in \cref{eqn:balanceboundaugreg3}.
However, it might be the case that~$ X $ has less than two missing debits.
In this case the balance of~$ X $ will be increased by additional credits.

\paragraph{Organization of the Proof.}
We start with the proof of two missing debits for singletons in \cref{lemma:lowDebs} (in \cref{lemma:lowDebs}\cref{lemma:lowDebsD}, in particular);
this result concludes the proof of \cref{eqn:balanceboundsinglereg3}.

More generally, \cref{lemma:lowDebs} also applies to paths and identifies those cases in which paths have two missing debits as well.
\cref{lemma:lowDebs} also prepares the analysis of the special case in which a path~$ X $ has less than two missing debits:
we will argue that~$ X $ has exactly one missing debit and in \cref{lemma:maxdeg3Graphs} in \cref{sect:onemissingdebandanextracred} we prove that the extra debit is compensated by an additional credit to~$ X $.
Hence \cref{lemma:lowDebs,lemma:maxdeg3Graphs} conclude the proof of \cref{eqn:balanceboundaugreg3}.

\begin{lemma}
\label{lemma:lowDebs}
Assume that degrees are bounded by at most~$ \maxdeg=3 $.
Consider the creation step of component~$ X $, and assume that the \onetwomingreedy algorithm selects node~$ u $ with current degree~$ d(u) $.
Component~$X$ has two missing debits if
\begin{enumerate}[label=\alph*),topsep=0mm,noitemsep]
\item\label{lemma:lowDebsA}
we have~$ d(u)=1 $
\item\label{lemma:lowDebsB}
we have~$ d(u)=\maxdeg=3 $
\item\label{lemma:lowDebsD}
component~$ X $ is a singleton.
\end{enumerate}
\end{lemma}
\begin{proof}

Let $ X $ be created in step $ s $ when $ u $ is matched with, say, node $ v $.
If neither~$ u $ nor~$ v $ pays a debit, then two debits are missing and we are done.
Hence in the rest of the proof we may assume that at least one of~$ u $ and~$ v $ pays a debit, which we denote by~\dedge{x,w} for~$ x{\in}\{u,v\} $.

\smallskip

We prove \ref{lemma:lowDebsA}.
If $ d(u)=1 $ holds at step $ s $, then no~$ F $-edges are incident with~$ u $.
Hence~$ u $ has two missing debits by Definition~\ref{def:transfer}.

\smallskip

We prove \ref{lemma:lowDebsB}.
Since~\dedge{x,w} is a transfer, step $ s $ removes edge $ \edge{x,w} $ from \g and after step $ s $ the degree of~$ w $ is~$ d'(w)\leq1 $.
Observe that in step $ s $ at most two edges incident with $ w $ are removed.
Since the degree of~$ w $ is at least~$ d(w)\geq d(u)=3 $ before step~$ s $,
we have $ d'(w)=1 $ at step $ s+1 $.
But since endpoint~$ w $ is never selected, another \degreeone node~$ y\neq w $ is selected next in step~$ s+1 $.
Observe that the degree of $ y $ also drops from~$ d(y)=3 $ to~$ d'(y)=1 $ in step~$ s $, namely when incident edges~$ \edge{u,y}$ and~$\edge{v,y} $ are removed.

\begin{itemize}
\item 
If {\boldmath\bfseries$ y $ belongs to a component~$ Y\neq X $} other than~$ X $, then both~$ \edge{u,y}$ and~$\edge{v,y} $ are $ F $-edges.
But~$ \edge{u,y}$ and~$\edge{v,y} $ are not debits, since $ u,v$ and~$y $ are~\m-covered. 
Hence we have found a missing debit for each of~$ u$ and~$v $.
\item 
If {\boldmath\bfseries$ y $ is a node of $ X $} (note that in this case~$ X $ must be a path), then nodes~$ u,v$, and~$y $ form a triangle and one of~$ \edge{u,y}$ and~$\edge{v,y} $, say~$\edge{u,y}$, is an~$ F $-edge. 
Since both~$ u$ and~$y $ are~\m-covered, edge~$\edge{u,y}$ is not a transfer and
hence both~$ u$ and~$ y $ have a missing debit.
\end{itemize}

\smallskip

We prove \ref{lemma:lowDebsD}.
We may assume that~$ d(u)=2 $ holds at creation of singleton~$ X $, since \cref{lemma:lowDebsA,lemma:lowDebsB} apply in particular to singletons.
Observe that a debit is missing from node~$ u $.
In order to obtain a contradiction, we assume that this is the only missing debit from~$ X $.
Consequently, node~$ u $ pays exactly one debit, say to $ w_u $, and~$ v $ pays exactly two debits, say to~$ w_v$ and~$w_v'$.

\pagebreak

Since at creation of~$ X $ node~$ u $ has degree~$ d(u)=2 $ and is adjacent to~$v$, node~$u$ is adjacent to at most one of~$ w_v$ and~$w_v'$.
Since edges~$\edge{u,w_u}$, $\edge{v,w_v}$, and~$\edge{v,w_v'}$ are transfers, Definition~\ref{def:transfer} implies that the degrees of~$w_u$, $w_v$ and,~$w_v'$ are at most~$ 1 $ after step~$ s $.
We distinguish two cases.

\begin{itemize}
\item 
If~{\bfseries\boldmath$u$ is adjacent to neither~$ w_v$ nor~$w_v'$}, then~$w_u$, $w_v$, and~$w_v'$ have degree exactly~$ 1 $ after step~$ s $:
their degrees were at least the minimum degree of~$d(u) = 2$ before~\edge{u,v} was picked, and their degrees dropped to at most one afterwards.

\item 
Now consider the case that~{\bfseries\boldmath$u$ is adjacent to~$w_v$ or~$ w_v' $}, say~$w_u = w_v$ holds.
Then the degree of~$w_v'$ drops by at most one when edge~$\edge{u,v}$ is picked, and hence~$w_v'$ has degree exactly~$ 1 $ afterwards.
\end{itemize}

In both cases, no other degrees than those of~$ w_u,w_v $, and~$ w_v' $ dropped in step $ s $.
Since in step $ s+1 $ one of these endpoints has degree exactly one, one of~$ w_u$,~$w_v$, or~$w_v' $ is selected and matched next.
A contradiction, since path endpoints are never matched by \onetwomingreedy.
\qed
\end{proof}

\subsection{One Missing Debit and an Additional Credit}
\label{sect:onemissingdebandanextracred}
To finish the proof of \cref{thm:mingreedyreg} it remains to verify \cref{eqn:balanceboundaugreg3} for any given path~$ X $.
Recall that the balance of~$ X $ is at least~$ c_X-d_X\geq2-D_X=2-2\w{X} $.
To prove a balance of at least~$ c_X-d_X\stackrel{!}{\geq}2-2\w{X}+2 $ it suffices to show that two debits are missing for~$ X $.

\medskip

Since by \cref{obs:potentialTransfer} path~$ X $ is created in a step selecting a node~$ u $ of current degree~$ d(u)\geq2 $, we analyze the cases that in the creation step of~$ X $ we have~$ d(u)=2 $ or~$ d(u)=3 $.

In the latter case, two debits are missing by \cref{lemma:lowDebs}\cref{lemma:lowDebsB}.
Hence, for the rest of the proof we may assume that~$ d(u)=2 $ holds.
Observe that no~$ F $-edge is incident with~$ u $ when being selected, hence a debit is missing for~$ u $.

If an additional debit is missing for~$ X $, then we are done once again.
So assume from here on that the \emph{only} missing debit for~$ X $ is missing for~$ u $.
Then~$ d_X=D_X-1=2\w{X}-1 $ holds.
Rather than proving a contradiction, in \cref{lemma:maxdeg3Graphs} we show that~$ X $ receives at least \emph{three} credits:
we obtain~$ c_X-d_X\geq3-2\w{X}+1 $, as required in \cref{eqn:balanceboundaugreg3}.

\begin{lemma}
\label{lemma:maxdeg3Graphs}
Assume that degrees are bounded by at most~$ \maxdeg=3 $.
Let~$ X $ be a path which pays~$ d_X=2\w{X}-1 $ debits.
Then~$ X $ receives at least~$ c_X\geq3 $ credits.
\end{lemma}
\begin{proof}
We assume that~$ c_X<3 $ holds and show a contradiction.

By assumption and as a consequence of \cref{prop:numcreditsreg}, each endpoint of~$ X $ receives exactly one credit.
Moreover, we have already argued before \cref{lemma:maxdeg3Graphs} that the node~$ u $ selected to create~$ X $ has degree~$ d(u)=2 $ at creation, and that~$u$ has the only missing debit of~$X$. 
Consequently, all other \m-covered nodes of $ X $ have exactly
one debit, since~$ d_X=2\w{X}-1 $ holds.
Let~$ u $ be matched with node~$ v $.

\medskip
\pagebreak
First we consider the case that~$ X $ is a~\oneOverTwo-path with~{\boldmath$ m_X=1 $}.
Let~$w_u$ and~$w_v$ be the endpoints of~$X$ such that we have~$\edge{u,w_u}, \edge{v,w_v} \in \mopt$. 
Note that this implies~$w_u \neq w_v$.
Node~$v$ pays the only debit, say to the path endpoint~$y$.
Observe that in the creation step of~$ X $ the degrees of~$w_u,w_v$, and~$y$ all drop, and no other degrees drop.

If the degree of~$ w_v $ is still at least two after~$ X $ was created, then~$ w_v $ receives two credits over two still incident~$ F $-edges.
A contradiction to our assumption that~$ c_X<3 $ holds.

Hence we may assume from here on that after creation of~$ X $ the degree of~$ w_v $ is at most~$ d'(w)\leq1 $.
But when~\edge{u,v} is picked we have~$ d(u)=2 $, i.e.\ endpoint~$ w_v $ is not incident with~$ u $ and the degree of~$ w_v $ drops by at most one.
Consequently, endpoint~$w_v$ has degree exactly~$ d'(w)=1 $ after creation, which is the new minimum degree in the graph.
But since only the degrees of~$ w_u,w_v $, and~$ y $ dropped in the creation step, one of~$ w_u,w_v$ and~$ y $ is matched in the next step.
A contradiction is obtained since path endpoints are never matched.
The possible configurations are depicted in \cref{fig:degDropContradiction2}.

\begin{figure}[htbp!]
\centering
\begin{tikzpicture}

\node (u) {$ u $};
\node[right of=u] (v) {$ v $};
\node[left of=u,dashed] (w1) {$ w_u $};
\node[right of=v,dashed] (w2) {$ w_v $};
\draw[mg] (u) -- (v);
\draw[opt] (u) -- (w1);
\draw[opt] (v) -- (w2);
\node[above of=v,dashed] (vdebit) {$ y $};
\node[above of=w1] (w1credit) {};
\node[above of=w2] (w2credit) {};
\draw[transfer] (v) -- (vdebit);
\draw[transfer] (w1credit) -- (w1);
\draw[transfer] (w2credit) -- (w2);

\node (u) at ($(u)+(5,0)$) {$ u $};
\node[right of=u] (v) {$ v $};
\node[left of=u] (w1) {$ w_u $};
\node[right of=v,dashed] (w2) {$ w_v $};
\draw[mg] (u) -- (v);
\draw[opt] (u) -- (w1);
\draw[opt] (v) -- (w2);
\node[above of=w2] (w2credit) {};
\draw[transfer] (w2credit) -- (w2);
\draw[transfer] (v) edge[bend left] (w1);
\end{tikzpicture}
\caption[\onetwomingreedy: Picking an ``End Edge'' of a Path (I)]{
Creating a~\oneOverTwo-path which pays one debit and receives exactly two credits:
dashed endpoints have degree one after $ u $ and~$ v $ are matched
(not all edges are drawn)
}
\label{fig:degDropContradiction2}
\end{figure}
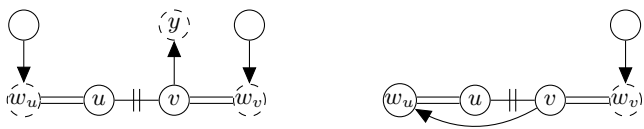

\medskip

Now consider the case that path~$ X $ has~{\boldmath$ \w{X}\geq2 $} edges of~\m.
After edge~\edge{u,v} is picked in the creation step of~$ X $, at least one path endpoint~$ w $ of~$ X $ is still connected with its unique~\mopt-neighbor, call it~$ x $.
We consider the step when~$x$ becomes matched, say with its~\m-neighbor~$ x' $.
In this step, we denote the current degree of a node~$ z $ as~$ d(z) $.

Recall that~$ x $ pays a debit, say to endpoint~$ y $.
Since~$w$ and~$ y $ are endpoints and are thus never matched, both~$ w $ and~$ y $ are not yet isolated.
Thus node~$x$ is still connected with~$ x',w $, and~$ y $ and has degree at least~$ d(x)\geq3 $.

Also, we may assume that~$w$ has degree at most~$ d(w)\leq2 $.
To see this, assume that the degree of~$ w $ is still~$\maxdeg = 3$ and observe that the two~$ F $-edges incident with~$w$ become credits, which contradicts our assumption that~$ c_X<3 $ holds.

Moreover, we argue that the degree of~$ w $ is exactly~$ d(w)=2 $.
Why?
Recall that~$ u $ has the only missing debit from~$ X $.
Hence~$x'$ pays a debit, i.e.\ node~$ x' $ is adjacent to an endpoint of an path.
Since~$ x' $ is also adjacent to~$ x $, node~$ x' $ has degree at least~$ d(x')\geq2 $.
But the degree of~$ x' $ is not larger than the degree~$ d(w) $ of~$ w $, i.e.\ it is at most~$ d(x')\leq d(w)\leq2 $.
Therefore the degree of~$ x' $ is exactly~$ d(x')=2 $.
Consequently, endpoint~$ w $ has degree exactly~$ d(w)=2 $ as well.

All still present neighbors of~$x$ and~$x'$ are endpoints of paths, and only their degrees drop when edge~\edge{x',x} is picked.
The possible configurations are depicted in \cref{fig:degDropContradiction}.
We distinguish two cases.

\begin{itemize}
\item 
If~{\boldmath\bfseries$x'$ is not adjacent to~$w$}, then the degree of~$ w $ drops by exactly one to exactly (the three leftmost configurations in \cref{fig:degDropContradiction}).
Hence either~$w$ or another path endpoint is selected in the next step.
A contradiction, since a path endpoint is never matched.

\item 
Lastly, assume that~{\bfseries\boldmath$x'$ and~$w$ are adjacent} (the rightmost configuration in \cref{fig:degDropContradiction}).
Then~$w$ becomes isolated in that step.
We consider the recipient~$y$ of the debit from~$x$.
Since~$ x' $ is adjacent to~$ w $ and~$ x $ and has degree exactly~$ d(x')=2 $, endpoint~$ y $ is not adjacent to~$ x' $.
Therefore the degree of~$ y $ drops by exactly one.
Since the degree of~$ y $ drops from at least~$ d(y)\geq2 $ to at most~$ 1 $, by \cref{def:transfer} we get that the degree of~$y$ is exactly~$ 1 $ in the next step.
Furthermore, endpoint~$ y $ is now the only \degreeone node.
Hence~$ y $ is selected and matched next.
A contradiction, since~$ y $ is a path endpoint.
\end{itemize}
\qed
\end{proof}

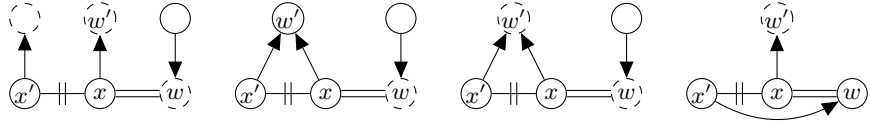
\begin{figure}[htbp!]
\centering
\begin{tikzpicture}

\node (u') {$ x' $};
\node[right of=u'] (v') {$ x $};
\node[right of=v',dashed] (w) {$ w $};
\draw[mg] (u') -- (v');
\draw[opt] (v') -- (w);
\node[above of=u',dashed] (u'debit) {};
\node[above of=v',dashed] (v'debit) {$ w' $};
\node[above of=w] (wcredit) {};
\draw[transfer] (u') -- (u'debit);
\draw[transfer] (v') -- (v'debit);
\draw[transfer] (wcredit) -- (w);

\node (u') at ($(w)+(1,0)$) {$ x' $};
\node[right of=u'] (v') {$ x $};
\node[right of=v',dashed] (w) {$ w $};
\draw[mg] (u') -- (v');
\draw[opt] (v') -- (w);
\node[] (u'v'debit) at ($(u')+(.5,1)$) {$ w' $};
\node[above of=w] (wcredit) {};
\draw[transfer] (u') -- (u'v'debit);
\draw[transfer] (v') -- (u'v'debit);
\draw[transfer] (wcredit) -- (w);

\node (u') at ($(w)+(1,0)$) {$ x' $};
\node[right of=u'] (v') {$ x $};
\node[right of=v',dashed] (w) {$ w $};
\draw[mg] (u') -- (v');
\draw[opt] (v') -- (w);
\node[dashed] (u'v'debit) at ($(u')+(.5,1)$) {$ w' $};
\node[above of=w] (wcredit) {};
\draw[transfer] (u') -- (u'v'debit);
\draw[transfer] (v') -- (u'v'debit);
\draw[transfer] (wcredit) -- (w);

\node (u') at ($(w)+(1,0)$) {$ x' $};
\node[right of=u'] (v') {$ x $};
\node[right of=v'] (w) {$ w $};
\draw[mg] (u') -- (v');
\draw[opt] (v') -- (w);
\node[above of=v',dashed] (v'debit) {$ w' $};
\draw[transfer] (v') -- (v'debit);
\draw (u') edge[transfer,bend right] (w);
\end{tikzpicture}
\caption[\onetwomingreedy: Picking an ``End Edge'' of a Path (II)]{
Removing the last \mopt-edge at an end of a path:
dashed endpoints have degree one after $ x $ and~$ x' $ are matched
(not all edges are drawn)
}
\label{fig:degDropContradiction}
\end{figure}

\section{Maximum Degree $ \maxdeg\geq4 $}
\label{sect:onetwomingreedy4}
In this section we show
\begin{theorem}
\label{thm:mingreedy4}
The \onetwomingreedy algorithm achieves approximation ratio at least~$ \frac{\maxdeg-1}{2\maxdeg-3} $ for graphs of maximum degree~$ \maxdeg\geq4 $.
\end{theorem}

In the proof we have to verify bounds \cref{def:edgedeb}, \cref{eqn:mintwocreds}, \cref{eqn:balanceboundsinglereg}, and \cref{eqn:balanceboundaugreg}.
First, we recall these bounds in an overview of how we refine the analysis applied for~$ \maxdeg=3 $.
In particular, the overview is structured so as to highlight the two main concepts used in our analysis.
\begin{itemize}
\item
Canceling Transfers

Since~\maxdeg is now larger than~$ 3 $, a path endpoint might receive more credits than for~$ \maxdeg=3 $.
However, extra transfers degrade the balance of their source components.

To annihilate this drawback we \emph{cancel} extra transfers: certain transfers as per \cref{def:transfer} will no longer move coins.
The discussion is found in \cref{sect:cancel}.

In particular, we carefully cancel transfers such that \cref{prop:numcreditsreg} still applies:
each endpoint of a path~$ X $ still receives at least one credit.
Consequently, bound \cref{eqn:mintwocreds} of at least
\begin{align*}
c_X\geq2\tag{\ref{eqn:mintwocreds}}
\end{align*}
credits to~$ X $ will be satisfied.
\item 
Donations

We discuss how to make a path~$ X $ balanced.
Assuming that \cref{eqn:mintwocreds,def:edgedeb} hold, i.e.\ that the minimum balance of~$ X $ is~$ c_X-d_X\geq2-D_X $, observe that
\begin{align*}
c_X-d_X\geq2-D_X+2(\maxdeg-2)\tag{\ref{eqn:balanceboundaugreg}}
\end{align*}
holds if~$ X $ has at least~$ 2(\maxdeg-2) $ missing debits.
However, we will see that the number of missing debits from~$ X $ might smaller, even as small as a constant.

To increase the balance of~$ X $ we introduce a \emph{donation} to~$ X $, which will compensate excessively payed debits from~$ X $.
Summing up, path~$ X $ will have to pay~$ 2(\maxdeg-2) $ debits less than maximum, as desired.
Like transfers, donations are edges in~$ F $.
The discussion is found in \cref{sect:donations}.
\end{itemize}

The balance of the source component~$ Y $ of a donation is reduced, since coins have to be payed.
Is~$ Y $ balanced?
If~$ Y $ is a singleton, then we show that both nodes of~$ Y $ pay at most
\begin{align*}
d_Y\leq 2(\maxdeg-2)\tag{\ref{eqn:balanceboundsinglereg}}
\end{align*}
coins even if a donation must be payed.
For a path~$ Y $ we have to verify \cref{def:edgedeb} for each~\m-edge~\edge{y,y'} of~$ Y $, i.e.\ we have to show that nodes~$ y $ and~$ y' $ pay at most
\begin{align*}
2(\maxdeg-2)\tag{\ref{def:edgedeb}}
\end{align*}
coins even if a donation must be payed.

\paragraph{Organization of the Proof.}
In \cref{sect:cancel} we introduce transfer cancellations and verify bound \cref{eqn:mintwocreds}.
Thereafter we discuss donations in \cref{sect:donations}.
In \cref{sect:maxnumcoins} we develop bounds on the number of coins payed by edges with an outgoing donation.
Finally, in \cref{sect:mostcomps} we show that all components are balanced by proving \cref{def:edgedeb}, \cref{eqn:balanceboundsinglereg}, and \cref{eqn:balanceboundaugreg}.

\subsection{Canceling Extra Transfers}
\label{sect:cancel}
Unlike for~$ \maxdeg=3 $ where each path endpoint has at most two incident~$ F $-edges and therefore receives at most two credits, for~$ \maxdeg\geq4 $ we will see that a path endpoint might receive up to three credits.
But extra credits decrease the balances of their source components.
In order to push up their balances, we show how to cancel extra credits in \cref{def:cancel}.
Thereafter we verify \cref{eqn:mintwocreds} in \cref{prop:newminmaxcred}.

\paragraph{Example.}
We prepare an example of an endpoint~$ w $ which receives an extra credit.
Assume that~$ w $ receives~$ k $ credits from nodes matched in steps~$ 1,\dots,s $:
in any step~$ s+l $ with~$ l\geq1 $ we say that~$ w $ \emph{already receives}~$ k $ credits.

Consider \cref{fig:threecreds}, where \onetwomingreedy begins by creating the~$ \maxdeg-5 $ many~\oneOverTwo-paths drawn above path~$ X $, then proceeds to pick the two edges of~$ X $ in steps~$ \maxdeg-4 $ and~$ \maxdeg-3 $, and eventually picks the two singletons drawn left and right of~$ X $.
Endpoint~$ w $ receives three credits in total:
in step~$ \maxdeg-3 $ node~$ w $ already receives two credits from the nodes matched to create path~$ X $, and afterwards~$ w $ receives an additional credit.
\begin{figure}[htbp!]
\centering
\begin{tikzpicture}
\node[ultra thick] (u) {};
\node[right=of u] (v) {};
\node[right=of v,ultra thick]  (u') {};
\node[right=of u'] (v') {};
\node[left=of u] (w) {};
\node[right=of v'] (w') {};

\node[ultra thick] at ($ (w)+(-1.5,-.75) $) (sl1) {};
\node at ($ (w)+(-1.5,+.75) $) (sl2) {};

\node[ultra thick] at ($ (w')+(1.5,-.75) $) (sr1) {};
\node at ($ (w')+(1.5,+.75) $) (sr2) {};

\node at ($ (u)+(0,2.5) $) (bpl1) {$ w $};
\node[above=of bpl1] (bpl2) {};
\node[ultra thick] at ($ (bpl2)+(-.75,1) $) (bpl3) {};
\node[right=of bpl3] (bpl4) {};

\node at ($ (u')+(0,2.5) $) (bpr1) {};
\node[above=of bpr1]  (bpr2) {};
\node[ultra thick] at ($ (bpr2)+(-.75,1) $) (bpr3) {};
\node[right=of bpr3]  (bpr4) {};

\draw[opt,mg] (sl1) -- node[right=1mm,alabel] {\rotatebox{90}{\scriptsize $\!\!\maxdeg{-}2 $}} (sl2);
\draw[opt,mg] (sr1) -- node[right=1mm,alabel] {\rotatebox{90}{\scriptsize $ \!\!\maxdeg{-}1 $}} (sr2);
\draw[opt] (w) -- (u);
\draw[opt] (v) -- (u');
\draw[opt] (v') -- (w');
\draw[opt] (bpl1) -- (bpl2);
\draw[opt] (bpl3) -- (bpl4);
\draw[opt] (bpr1) -- (bpr2);
\draw[opt] (bpr3) -- (bpr4);

\draw[mg] (u) -- node[below=1mm,alabel] {{\scriptsize$ \!\maxdeg{-}4 $}} (v);
\draw[mg] (u') -- node[below=1mm,alabel] {{\scriptsize$ \!\maxdeg{-}3 $}} (v');
\draw[mg] (bpl2) -- node[midway,xshift=-3mm,yshift=-3mm,alabel] {\rotatebox{-50}{\scriptsize $\!\! 1 $}} (bpl3);
\draw[mg] (bpr2) -- node[midway,xshift=-3mm,yshift=-3mm,alabel] {\rotatebox{-50}{\scriptsize$\!\! \maxdeg{-}5 $}} (bpr3);

\draw[transfer,->] (bpl2) -- (bpl4);
\draw[transfer,->] (bpr2) -- (bpr4);
\draw[transfer,->] (sl1) -- (w);
\draw[transfer,->] (sl2) -- (w);
\draw[transfer,->] (sr1) -- (w');
\draw[transfer,->] (sr2) -- (w');
\draw[transfer,->] (u) -- (bpl1);
\draw[transfer,->] (u) -- (bpr1);
\draw[transfer,->] (v) -- (bpl1);
\draw[transfer,->] (v) edge[bend right=0] (bpr1);
\draw[transfer,->] (v') -- (bpl1);
\draw[transfer,->] (v') -- (bpr1);

\draw (sl1) edge[bend right=15] (u);
\draw (sl2) edge[bend left=15] (u);

\draw (sr1) edge[bend left=15] (v');
\draw (sr2) edge[bend right=15] (v');

\draw (u) edge[bend right=60] (u');

\node[alabel] at ($ (bpl2)+(1.5,0) $) {$ \dots $};

\node[alabel] at ($ (u)+(.4,.75) $) {\rotatebox{-20}{$ \dots $}};
\node[alabel] at ($ (v)+(0,.5) $) {\rotatebox{0}{$ \;\dots $}};
\node[alabel] at ($ (v')+(-1.15,1) $) {\rotatebox{45}{$ \;\dots $}};

\draw [decorate,decoration={brace,amplitude=1.666em}]
($ (w')+(0,-1.333) $) -- ($ (w)+(0,-1.333) $) node [alabel,midway,yshift=-2.5em] 
{$ X $};
\end{tikzpicture}
\caption[\onetwomingreedy: An Endpoint With an Extra Credit]{Endpoint~$ w $ receives three credits from path~$ X $
(\m-edges are numbered to indicate the order in which they are picked, where \onetwomingreedy selects fat nodes)}
\label{fig:threecreds}
\end{figure}
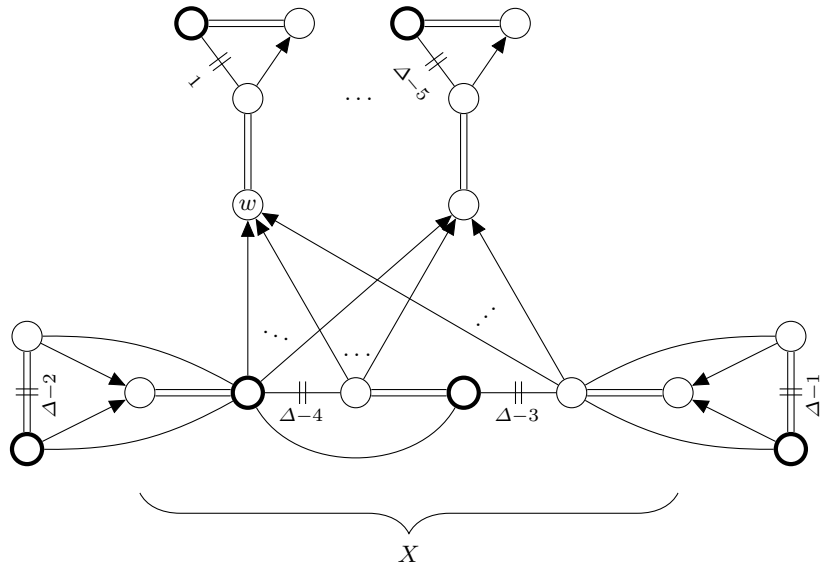

However, path~$ X $ is not balanced.
Why?
Path~$ X $ pays three debits to \emph{each} of the~$ \maxdeg-5 $ topmost~\mbox{\oneOverTwo-paths}, and receives only four credits from the two singletons.
Hence the balance of~$ X $ is as small as~$ c_X-d_X=4-3(\maxdeg-5)$, whereas balance at least~$c_X-d_X\geq2-2\w{X}(\maxdeg-2)+2(\maxdeg-2)=2-2(\maxdeg-2) $ is required by \cref{eqn:balanceboundaugreg}.

\paragraph{The Definition.}
Therefore we cancel an extra credit to~$ w $, thereby limiting the number of credits to~$ w $ to at most two (like for~$ \maxdeg=3 $).
To prepare the definition we first analyze the exact ``configuration'' of a path endpoint with more than two credits.
\pagebreak
\begin{lemma}
\label{prop:maxthreecreds}
As a consequence of \cref{def:transfer} of transfers,
\begin{enumerate}[noitemsep,topsep=0mm,label=\alph*)]
\item\label{prop:maxthreecreds_bound} each path endpoint receives at most three credits, and
\item\label{prop:maxthreecreds_situation} path endpoint~$ w $ receives three credits if and only if the the following holds:
the degree of~$ w $ drops
\begin{enumerate}[topsep=0mm,noitemsep,label=\roman*.]
\item\label{prop:maxthreecreds_situation:3to1}
from~$ d(w)=3 $ to~$ d(w)=1 $ when two incident~$ F $-edges are removed and 
\item\label{prop:maxthreecreds_situation:1to0}
from~$ d(w)=1 $ to~$ d(w)=0 $ when the last incident~$ F $-edge is removed.
\end{enumerate}
\end{enumerate}
\end{lemma}
\begin{proof}
First we argue that an endpoint~$ w $ receives \emph{less} than three credits if there is a step~$ s $ when the current degree of~$ w $ is~$ d(w)=2 $.
By \cref{def:transfer}, no~$ F $-edge of~$ w $ removed before step~$ s $ is a credit to~$ w $.
Moreover, in step~$ s $ at most two~$ F $-edges are still incident with~$ w $.

Consequently, if endpoint~$ w $ receives at least three credits, then~$ w $ never has current degree~$ d(w)=2 $.
This implies the following two facts.
First, the degree of~$ w $ in~\g is~$ d_\g(w)\geq3 $, since the degree of~$ w $ in~\g is~$ d_\g(w)\geq2 $ by \cref{obs:potentialTransfer}.
Secondly, there is a step when the current degree of~$ w $ is~$ d(w)=3 $:
the degree~$ d(w) $ of~$ w $ never equals two, hence we get that~$ d(w) $ drops from larger than~$ d(w)>2 $ to below~$ d(w)<2 $, which is only possible if~$ d(w) $ drops from exactly~$ d(w)=3 $, since in each step at most two edges incident with~$ w $ are removed.

\cref{prop:maxthreecreds_bound}
In the step when~$ d(w)=3 $ holds, by \cref{def:transfer} of transfers endpoint~$ w $ still has zero credits.
Hence only the remaining three incident edges can be transfer~$ F $-edges.
Thus~$ w $ receives at most~$ 3 $ credits.

\cref{prop:maxthreecreds_situation}
In the step when~$ d(w)=3 $ holds, all remaining edges incident with~$ w $ must be~$ F $-edges, since otherwise~$ w $ would receive less than three credits.
The statement follows, since we have already argued that~$ w $ never has degree~$ d(w){=}2 $.
\qed
\end{proof}

If endpoint~$ w $ receives the maximum of three credits, then we cancel the ``third'' credit, i.e.\ the one coming from a node getting matched in the step when the degree of~$ w $ drops from one to zero, cf. \cref{prop:maxthreecreds}\cref{prop:maxthreecreds_situation}\cref{prop:maxthreecreds_situation:1to0}
(Observe that path~$ X $ in \cref{fig:threecreds} is now balanced, since the debits from its rightmost~\m-covered node are canceled and therefore the balance~$ c_X-d_X\geq4-2(\maxdeg-5)>2-2(\maxdeg-2) $ is now strictly larger than \cref{eqn:balanceboundaugreg} requires.)

\begin{definition}
\label{def:cancel}
Let~$ w $ be a path endpoint with current degree~$ d(w)=1 $ and one incident~$ F $-edge~\edge{v,w}.
If~$ w $ already receives two credits, then we \emph{cancel} transfer~\dedge{v,w}, i.e.\ edge~\edge{v,w} does \emph{not} move coins.
\end{definition}

The final set of transfers is given by \cref{def:transfer,def:cancel}.
We are now ready to prove \cref{eqn:mintwocreds}.

\begin{lemma}
\label{prop:newminmaxcred}
\begin{enumerate}[label=\alph*)]
\item\label{prop:newminmaxcred:atmost}
Each path endpoint~$ w $ receives at most two credits.
An endpoint for which a credit was canceled receives exactly two credits.
\hspace*{-1.5mm}\item\label{prop:newminmaxcred:atleast}
Each path endpoint~$ w $ receives at least one credit.
Therefore \cref{eqn:mintwocreds} holds.
\end{enumerate}
\end{lemma}
\pagebreak
\begin{proof}
\cref{prop:newminmaxcred:atmost} is a direct consequence of \cref{prop:maxthreecreds,def:cancel}, since~$ w $ receives at most three credits, and if so then the third credit is canceled.
To prove \cref{prop:newminmaxcred:atleast} we only have to consider path endpoints with less than two credits.
In particular, we only have to study an endpoint~$ w $ for which no credit was canceled, since by \cref{prop:newminmaxcred:atmost} only such an endpoint can have less than two credits.
Now observe that \cref{prop:numcreditsreg} applies to~$ w $, no matter if a credit was canceled for~$ w $ or not:
by \cref{prop:numcreditsreg} endpoint~$ w $ receives at least one credit.
\qed
\end{proof}

\subsection{Moving Coins to Paths in Donations}
\label{sect:donations}
We prepare the discussion and introduce some notation, which will also be used in the remainder of \cref{sect:onetwomingreedy4}.
Let~$ X $ be a path and consider the creation step of~$ X $.
We denote the current degree of a node~$ x $ as~$ d(x) $.
The nodes matched in the creation step are called~$ u $ and~$ v $, where we assume that \onetwomingreedy selects~$ u $ and matches~$ u $ with neighbor~$ v $.
In the next step, we denote the current degree of~$ x $ as~$ d'(x) $, and we call the matched nodes~$ u' $ and~$ v' $, where we assume that~$ u' $ is selected and matched with neighbor~$ v' $.

\paragraph{Organization of this Section.}
We begin this section with an example of a path which pays many debits from nodes~$ u $ and~$ v $.
Thereafter, we sketch that nodes~$ u' $ and~$ v' $ are good candidates to compensate for the excessively payed debits from~$ u $ and~$ v $.
We conclude with the definition of donations---which move coins over~$ F $-edges---and a discussion of \emph{donation steps}, i.e.\ the steps matching source nodes of donations.

\paragraph{Example.}
Recall that to verify \cref{eqn:balanceboundaugreg} we would like to show that for each path at least~$ 2(\maxdeg-2) $ debits are missing.
However, even after canceling transfers a path might pay too many debits.
In particular, the example in \cref{fig:threecreds} can be changed slightly as follows, see \cref{fig:threecreds:single}:
like in \cref{fig:threecreds} the \onetwomingreedy algorithm first creates the topmost~$ \maxdeg-5 $ many~\oneOverTwo-path end creates the left and right singletons in the last two steps;
however, in steps~$ \maxdeg-4 $ and~$ \maxdeg-3 $ the algorithm creates a~\oneOverTwo-path~$ X $ and another singleton.
Path~$ X $ receives~$ c_{X}=4 $ credits and pays~$ d_{X}=2(\maxdeg-5) $ debits, i.e.\ its balance is~$ c_{X}-d_{X}=4-2(\maxdeg-5) $, whereas balance at least~$ c_{X}-d_{X}\geq2-2\w{X}(\maxdeg-2)+2(\maxdeg-2)=2 $ is required by \cref{eqn:balanceboundaugreg}.

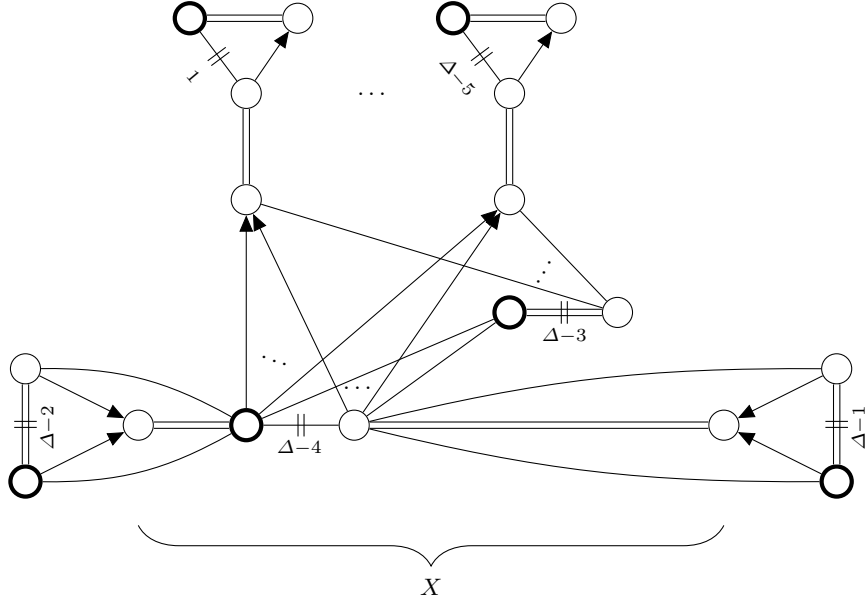
\begin{figure}[htbp!]
\centering
\begin{tikzpicture}

\node[ultra thick] (u) {};
\node[right=of u] (v) {};
\node[left=of u] (w) {};
\node[right=4.5 of v] (w') {};

\node[ultra thick] at ($ (w)+(-1.5,-.75) $) (sl1) {};
\node at ($ (w)+(-1.5,+.75) $) (sl2) {};

\node[ultra thick] at ($ (w')+(1.5,-.75) $) (sr1) {};
\node at ($ (w')+(1.5,+.75) $) (sr2) {};

\node at ($ (u)+(0,3) $) (bpl1) {};
\node[above=of bpl1] (bpl2) {};
\node[ultra thick] at ($ (bpl2)+(-.75,1) $) (bpl3) {};
\node[right=of bpl3] (bpl4) {};

\node at ($ (bpl1)+(3.5,0) $) (bpr1) {};
\node[above=of bpr1]  (bpr2) {};
\node[ultra thick] at ($ (bpr2)+(-.75,1) $) (bpr3) {};
\node[right=of bpr3]  (bpr4) {};

\draw[opt,mg] (sl1) -- node[right=1mm,alabel] {\rotatebox{90}{\scriptsize $\!\!\maxdeg{-}2 $}} (sl2);
\draw[opt,mg] (sr1) -- node[right=1mm,alabel] {\rotatebox{90}{\scriptsize $ \!\!\maxdeg{-}1 $}} (sr2);
\draw[opt] (w) -- (u);
\draw[opt] (v) -- (w');
\draw[opt] (bpl1) -- (bpl2);
\draw[opt] (bpl3) -- (bpl4);
\draw[opt] (bpr1) -- (bpr2);
\draw[opt] (bpr3) -- (bpr4);

\draw[mg] (u) -- node[below=1mm,alabel] {{\scriptsize$ \!\maxdeg{-}4 $}} (v);
\draw[mg] (bpl2) -- node[midway,xshift=-3mm,yshift=-3mm,alabel] {\rotatebox{-50}{\scriptsize $\!\! 1 $}} (bpl3);
\draw[mg] (bpr2) -- node[midway,xshift=-3mm,yshift=-3mm,alabel] {\rotatebox{-50}{\scriptsize$\!\! \maxdeg{-}5 $}} (bpr3);

\draw[transfer,->] (bpl2) -- (bpl4);
\draw[transfer,->] (bpr2) -- (bpr4);
\draw[transfer,->] (sl1) -- (w);
\draw[transfer,->] (sl2) -- (w);
\draw[transfer,->] (sr1) -- (w');
\draw[transfer,->] (sr2) -- (w');
\draw[transfer,->] (u) -- (bpl1);
\draw[transfer,->] (u) -- (bpr1);
\draw[transfer,->] (v) -- (bpl1);
\draw[transfer,->] (v) -- (bpr1);

\draw (sl1) edge[bend right=15] (u);
\draw (sl2) edge[bend left=15] (u);

\draw (sr1) edge[bend left=7] (v);
\draw (sr2) edge[bend right=7] (v);

\node[alabel] at ($ (bpl2)+(1.65,0) $) {$ \;\dots $};

\node[alabel] at ($ (u)+(.4,.85) $) {\rotatebox{-20}{$ \dots $}};
\node[alabel] at ($ (v)+(0,.5) $) {\rotatebox{0}{$ \;\,\dots $}};

\node[ultra thick] at ($ (bpr1)+(0,-1.5) $) (s2) {};
\node[right=of s2] (s1) {};

\draw[opt,mg] (s1) -- node[alabel,below=1mm] {\rotatebox{0}{\scriptsize$ \maxdeg{-}3 $}} (s2);

\draw (u) -- (s2);
\draw (v) -- (s2);

\draw[rounded corners=10pt] (s1) -- (bpl1);
\draw (s1) -- (bpr1);

\node[alabel] at ($ (s1)+(-1,.5) $) {\rotatebox{60}{$ \;\,\dots $}};

\draw [decorate,decoration={brace,amplitude=1.666em}]
($ (w')+(0,-1.333) $) -- ($ (w)+(0,-1.333) $) node [alabel,midway,yshift=-2.5em] 
{$ X $};
\end{tikzpicture}
\caption[\onetwomingreedy: A~\oneOverTwo-Path With too Many Debits]{A~\oneOverTwo-path~$ X $ which pays too many debits
(\m-edges are numbered to indicate the order in which they are picked, where \onetwomingreedy selects fat nodes)}
\label{fig:threecreds:single}
\end{figure}

\paragraph{Sketch.}
Hence we have to provide more coins to an unbalanced path~$ X $.
Therefore we focus on the creation step of~$ X $ as well as on the following step.
Recall that \cref{eqn:balanceboundaugreg} holds if nodes~$ u $ and~$ v $ pay~$ k=0 $ debits, since then at least~$ 2(\maxdeg-2) $ debits are missing in total.
If nodes~$ u $ and~$ v $ pay
\begin{align*}
k>0
\end{align*}
coins over debits, then it will turn out that~$ u' $ and~$ v' $ have missing debits, which can be used to compensate for excessively payed debits from~$ u $ and~$ v $.
Depending on whether~$ u' $ and~$ v' $ belong to~$ X $ or not we proceed as follows:
\begin{itemize}
\item
If~$ u' $ and~$ v' $ belong to~$ X $, then the numbers of missing debits for~$ u,v,u'$ and~$ v' $ will add up to at least~$ 2(\maxdeg-2) $:
the balance of~$ X $ is~$ c_X-d_X\geq 2-D_X+2(\maxdeg-2) $, as required in \cref{eqn:balanceboundaugreg}.
\item
If otherwise~$ u' $ and~$ v' $ belong to another component~$ Y\neq X $, then without~$ Y $ becoming unbalanced component~$ Y $ will be able to \emph{donate}~$ k $ coins to~$ u $ and~$ v $:
the~$ k $ debits payed by~$ u $ and~$ v $ are compensated and again the balance of~$ X $ is increased by at least~$ 2(\maxdeg-2) $ over~$ 2-D_X $.

We call the step in which~$ u' $ and~$ v' $ are matched the \emph{donation step of~$ X $}.
\end{itemize}

\paragraph{Degree-1 Endpoints After Creation.}
We argue that debits are missing for~$ u' $ and~$ v' $.
Recall that we need only consider the case that a debit leaves~$ u $ or~$ v $, say to endpoint~$ w $, since otherwise we have found sufficiently many missing debits for \cref{eqn:balanceboundaugreg} to hold.
\cref{def:transfer} of transfers requires that~$ w $ has degree at most~$ d'(w)\leq1 $ after the creation step of~$ X $.
The key to our proof is:

endpoint~$ w $ might have degree \emph{exactly}~$ d'(w)=1 $ after creation of~$ X $.

\begin{definition}
\label{def:deg1endpaftercreate}
Let~$ X $ be a path being created when \onetwomingreedy selects node~$ u $ and matches it with~$ v $.
Assume that~$ u $ or~$ v $ pays a debit, say to endpoint~$ w $.
If after creation of~$ X $ endpoint~$ w $ has degree exactly~$ d'(w)=1 $, then we say that \emph{a \degreeone endpoint exists after creation of~$ X $}.
\end{definition}
\pagebreak

\noindent
Assume that a \degreeone endpoint exists after creation of~$ X $, call it~$ w $.
Since endpoint~$ w $ will never be matched, the \onetwomingreedy algorithm selects another \degreeone node, namely node~$ u' $, next.
As desired, we have found missing debits for~$ u' $:
since~$ u' $ has degree~$ d'(u')=1 $ when being selected, node~$ u' $ pays zero debits by \cref{def:transfer}.

\medskip

In \cref{lemma:createpathdeg1} we identify the ``configurations'' in which a \degreeone endpoint exists after creation of a path.
In particular, in \cref{lemma:createpathdeg1}\cref{lemma:createpathdeg1:deg2} we identify the exact configurations in which \emph{no} \degreeone endpoint exists after creation.

\begin{lemma}
\label{lemma:createpathdeg1}
Consider the creation step of path~$ X $, where nodes~$ u $ and~$ v $ are matched and \onetwomingreedy selects~$ u $ with current degree~$ d(u) $.
Assume that one of~$ u $ or~$ v $ pays a debit~$ t $, say to node~$ w $.
\begin{enumerate}[topsep=0mm,noitemsep,label=\alph*)]
\item\label{lemma:createpathdeg1:deg3}
If~$ d(u){\geq}3 $ holds, then a \degreeone endpoint exists after creation of~$ X $.
\item\label{lemma:createpathdeg1:deg2}
If~$ d(u){=}2 $ holds, then a \degreeone endpoint exists after creation of~$ X $, unless the following holds (see \cref{fig:debsnodeg1endpoint} for an illustration of this exception):
\begin{enumerate}[topsep=0mm,noitemsep,label=\roman*.]
\item\label{lemma:createpathdeg1:w0}
after creation of~$ X $ endpoint~$ w $ has degree~$ d'(w)=0 $,
\item\label{lemma:createpathdeg1:w2}
at creation of~$ X $ endpoint~$ w $ has degree~$ d(w)=2 $,

\item\label{lemma:createpathdeg1:nodebu}
no debit leaves~$ u $,
\item\label{lemma:createpathdeg1:debv}
debit~$ t $ leaves~$ v $, i.e.\ we have~$ t=\dedge{v,w} $,
\item\label{lemma:createpathdeg1:uw}
we have~$ \edge{u,w}\in\mopt $,
\item\label{lemma:createpathdeg1:wi}
at creation of~$ X $ all other endpoints~$ w_1,\dots,w_k $ neighboring~$ v $ (i.e.\ we have~$ w\notin\{w_1,\dots,w_k\} $) have degree at least three, and
\item\label{lemma:createpathdeg1:onedeb}
node~$ v $ pays no other debits besides~$ t $.
\end{enumerate}

\end{enumerate}
\begin{figure}
\centering
\begin{tikzpicture}
\node (w) {$ w $};
\node[below=of w,ultra thick] (u) {$ u $};
\node at ($ (u)+(1.5,.75) $) (v) {$ v $};
\node at ($ (v)+(0,1) $) (w1) {$ w_1 $};
\node at ($ (v)+(+2,1) $) (wr) {$ w_k $};

\draw[opt] (w) -- (u);
\draw[opt] (v) -- ($ (v)+(1.5,0) $);

\draw[mg] (u) -- node[below=2mm,alabel] {$ X $} (v);

\draw[transfer,->] (v) -- (w);

\draw (v) -- (w1);
\draw (v) -- (wr);

\draw (w1) -- ($ (w1)+(-.5,1.5) $);
\draw (w1) -- ($ (w1)+(+.5,1.5) $);
\draw (wr) -- ($ (wr)+(-.5,1.5) $);
\draw (wr) -- ($ (wr)+(+.5,1.5) $);

\node[above=of w1,alabel] {$ \;\dots $};
\node[above=of wr,alabel] {$ \;\dots $};
\node[right=.4 of w1,alabel] {$ \;\dots $};

\node[right=1.5 of v,alabel] {$ \dots $};
\end{tikzpicture}
\caption[\onetwomingreedy: No \degreeone Endpoint Exists after Creation]{No \degreeone endpoint exists after creation of path~$ X $ (the~\mopt-neighbor of~$ v $ could be one of the~$ w_i $)}
\label{fig:debsnodeg1endpoint}
\end{figure}
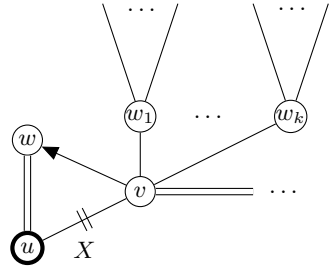
\end{lemma}
\begin{proof}

By \cref{def:transfer}, endpoint~$ w $ has degree at most~$ d'(w)\leq1 $ in the step after creation of~$ X $.
Observe that before creation of~$ X $ endpoint~$ w $ is not isolated, since~$ t $ is an incoming transfer and hence edge~$ t $ is still incident with~$ w $.

\cref{lemma:createpathdeg1:deg3}
If~$ d(u)\geq3 $ holds, then the degree of~$ w $ drops from at least~$ d(w)\geq3 $ to at least~$ d'(w)\geq1 $, since at most two edges incident with~$ w $ are removed.
Hence~$ w $ has degree exactly~$ d'(w)=1 $ after creation of~$ X $.

\cref{lemma:createpathdeg1:deg2}
Now assume that~$ d(u)=2 $ holds.
Since we have~$ d'(w)\leq1 $ after creation of~$ X $, to prove the statement we only need to study the case that no \degreeone endpoint exists after creation, i.e.\ that endpoint~$ w $ has degree~$ d'(w)=0 $ after creation.
Observe that \cref{lemma:createpathdeg1:w0}\ holds and it remains to verify \cref{lemma:createpathdeg1:w2,lemma:createpathdeg1:nodebu,lemma:createpathdeg1:debv,lemma:createpathdeg1:uw,lemma:createpathdeg1:wi,lemma:createpathdeg1:onedeb}

\begin{itemize}
\item[\hspace*{-1cm}\cref{lemma:createpathdeg1:w2}]
Endpoint~$ w $ has degree~$ d(w)=2 $ at creation, since the degree is~$ d(w)\geq2 $ by \cref{obs:potentialTransfer} and if the degree was at least~$ d(w)\geq3 $ then it could not drop to~$ d'(w)=0 $ after creation, as is fact by \cref{lemma:createpathdeg1:w0}
\item[\hspace*{-1cm}\cref{lemma:createpathdeg1:nodebu}]
Since~$ u $ has degree~$ d(u)=2 $ at creation of~$ X $, node~$ u $ is not incident with an~$ F $-edge, as would be required in order to pay a debit.
\item[\hspace*{-1cm}\cref{lemma:createpathdeg1:debv}]
Since debit~$ t $ does not leave~$ u $ by \cref{lemma:createpathdeg1:nodebu}, debit~$ t $ leaves~$ v $.
\item[\hspace*{-1cm}\cref{lemma:createpathdeg1:uw}]
Since in the creation step of~$ X $ the degree of~$ w $ drops by two, an edge connecting~$ w $ with each of~$ u $ and~$ v $ is removed, and one of both edges must be the~\mopt-edge incident with~$ w $.
But we have~$ t=\dedge{v,w} $, i.e.\ edge~\edge{v,w} is an~$ F $-edge.
We obtain~$ \edge{u,w}\in\mopt $.
\item[\hspace*{-1cm}\cref{lemma:createpathdeg1:wi}]
Observe that for each~$ w_i $ we have~$ d'(w_i)\geq1 $ after creation of~$ X $, since before creation of~$ X $ we have~$ d(w_i)\geq d(u)=2 $ and~$ w_i $ is not connected with~$ u $.
Since we assumed that no \degreeone endpoint exists after creation of~$ X $ and~$ w_i $ is not isolated then, we obtain~$ d'(w_i)\geq2 $.
But an edge incident with~$ w_i $ is removed in the creation step of~$ X $.
Hence we get~$ d(w_i)\geq3 $.
\item[\hspace*{-1cm}\cref{lemma:createpathdeg1:onedeb}]
This follows from \cref{lemma:createpathdeg1:wi}
Why?
If all~$ w_i $ have degree at least three at creation of~$ X $, then no~$ w_i $ receives a transfer from~$ v $, since the degree of each~$ w_i $ drops by exactly one (recall that~$ u $ is adjacent only with~$ v $ and~$ w $).
Consequently, endpoint~$ w $ receives the only transfer from~$ v $.
\qed
\end{itemize}

\end{proof}

\paragraph{The Definition of Donations.}
Donations will only be needed to move coins to paths for which a \degreeone endpoint exists after creation.
Hence in the rest of this section we focus on this case.
(In \cref{lemma:nonbip4exceptional} we analyze the case that no \degreeone endpoint exists after creation.)

Assume that a \degreeone endpoint exists after creation of path~$ X $.
Since~$ u' $ is selected with current degree~$ d'(u')=1 $ and hence pays zero debits, node~$ u' $ is a preferred candidate to compensate the debits payed by~$ u $ and~$ v $.

To prepare the definition of donations, in \cref{lemma:u'props} we show that an edge~$ e $ connecting~$ u' $ with~$ u $ or~$ v $ is removed in the creation step of~$ X $.
In particular, if a donation needs to be given to~$ X $ from another component~$ Y\neq X $, then~$ e $ is an edge in~$ F $.

\begin{lemma}
\label{lemma:u'props}
Assume that a \degreeone endpoint exists after creation of path~$ X $.
The node~$ u' $ selected next by \onetwomingreedy is
\begin{enumerate}[topsep=0mm,noitemsep,label=\alph*)]
\item\label{lemma:u'props:mopt}
an~\mopt-neighbor of~$ u $ or~$ v $ in~$ X $ or
\item\label{lemma:u'props:f}
an~$ F $-neighbor of~$ u $ or~$ v $ in a component~$ Y\neq X $.
\end{enumerate}
\end{lemma}
\begin{proof}
After creation of~$ X $ let~$ w $ be a \degreeone endpoint.
Since at creation of~$ X $ all degrees are at least two, thereafter both \degreeone nodes~$ w $ and~$ u' $ lost an edge connecting to~$ u $ or~$ v $.
Hence node~$ u' $ was connected to~$ u $ or~$ v $ by an~\mopt-edge or by an~$ F $-edge.

If~$ u' $ was connected by an~\mopt-edge, then node~$ u' $ belongs to~$ X $.
This proves \cref{lemma:u'props:mopt}.

To prove \cref{lemma:u'props:f}, assume that~$ u' $ is not an~\mopt-neighbor of~$ u $ or~$ v $.
Then~$ u' $ was connected only by an~$ F $-edge with~$ u $ or~$ v $.
It remains to show that~$ u' $ is not a node of~$ X $.
To see this, observe that after creation of~$ X $ all~\m-covered nodes of~$ X $ (other than the~\mopt-neighbors of~$ u $ and~$ v $) have degree at least two and recall that~$ u' $ is selected when it has degree~$ d'(u')=1 $.
\qed
\end{proof}

\begin{definition}
\label{def:donation}
Let nodes~$ u $ and~$ v $ be matched in the creation step of path~$ X $, where \onetwomingreedy selects node~$ u $ with current degree~$ d(u) $.
Assume that~$ u $ and~$ v $ pay~$ k>0 $ debits.
Let nodes~$ u' $ and~$ v' $ of component~$ Y\neq X $ be matched in the donation step of~$ X $, where \onetwomingreedy selects~$ u' $ with current degree~$ d'(u'){=}1 $.
Denote the~$ F $-edge connecting~$ u' $ with~$ u $ or~$ v $ as~\edge{u',x} for~$ x\in\{u,v\} $.
\begin{itemize}[noitemsep]
\item If~$ d(u)=2 $ holds, then edge~$ \edge{u',x}=\edge{u',v} $ is called a \emph{static donation} and moves~$ \maxdeg{-}3 $ coins to~$ v $.
\item If~$ d(u)\geq3 $ holds, then edge~\edge{u',x} is called a \emph{dynamic donation} and moves~$ k $ coins to~$ x $.
\end{itemize}
Donation~$ \edge{u',x} $ is denoted as a bold directed edge~$ \don{u',x} $ from the paying node~$ u' $ of~$ Y $ to the receiving node~$ x $ of~$ X $.
\end{definition}

\paragraph{Soundness of the Definition.}
Recall that our proof of \cref{eqn:balanceboundaugreg} crucially relies on bound \cref{def:edgedeb}, i.e.\ on the (yet unproven) fact that each~\m-edge of path~$ X $ pays at most~$ 2(\maxdeg{-}2) $ coins over outgoing debits and donations:
in total, all~\m-edges of~$ X $ pay at most~$ D_X=\w{X}\cdot2(\maxdeg-2) $ coins.
Since the endpoints of~$ X $ receive at least~$ c_X\geq2 $ credits by \cref{eqn:mintwocreds}, the balance of~$ X $ is at least~$ c_X-d_X\geq2-D_X $.
Our approach to use a donation to compensate \emph{all} coins payed by nodes~$ u $ and~$ v $ implies that the balance of~$ X $ is increased by at least~$ 2(\maxdeg-2) $:
the balance is at least~$ c_X-d_X\geq2-D_X +2(\maxdeg-2)$, as required in \cref{eqn:balanceboundaugreg}.

But to develop our \cref{def:donation} of donations we only considered coins payed by~$ u $ and~$ v $ over debits:
if~$ u $ or~$ v $ pays a donation whose coins are \emph{not} compensated, then our approach to verify \cref{eqn:balanceboundaugreg} fails.
However, \cref{prop:donationsaresafe}\cref{prop:donationsaresafe:disjoint} implies that neither~$ u $ nor~$ v $ pays a donation, i.e.\ our \cref{def:donation} of donations is consistent with our approach to analyze the creation step of~$ X $.
Moreover, \cref{prop:donationsaresafe}\cref{prop:donationsaresafe:exclusive} shows that coins donated to~$ X $ need not be shared with another path.

\begin{lemma}
\label{prop:donationsaresafe}
Let set~$ C=\{s_1,s_2,\dots\} $ with~$ s_1<s_2<\dots $ be the set of path creation steps, and set~$ D=\{s_1{+}1,\,s_2{+}1,\dots\} $ the set of donation steps.

\noindent
\begin{enumerate}[label=\alph*)]
\item\label{prop:donationsaresafe:disjoint}
We have~$ C\cap D=\emptyset $.
\hspace*{-1.5mm}\item\label{prop:donationsaresafe:exclusive}
The node selected in donation step~$ s_i{+}1 $ donates coins to exactly one node, which was matched in the creation step~$ s_i $.
\end{enumerate}
\end{lemma}
\begin{proof}
\cref{prop:donationsaresafe:disjoint}
Consider an arbitrary donation step~$ s_i{+}1 $.
Assume that \onetwomingreedy selects node~$ u_i' $, and recall that~$ u_i' $ is selected with current degree~$ d'(u_i')=1 $.
However, \cref{obs:potentialTransfer} shows that in each path creation step~$ s_j $ a node~$ u_j $ of current degree at least~$ d(u_j)\geq2 $ is selected.
We get~$ s_i{+}1\notin C $ and the statement follows.

\cref{prop:donationsaresafe:exclusive}
It suffices to show that~$ s_{i-1}{+}1<s_i $ holds, since then creation step~$ s_i $ is the only creation step between donation step~$ s_i{+}1 $ and the previous donation step~$ s_{i-1}{+}1 $.
The previous donation step~$ s_{i-1}{+}1 $ selects a node of degree one, whereas step~$ s_i $ selects a node of degree at least two in order to create a path.
Therefore~$ s_{i-1}{+}1\neq s_i $ holds.
Using~$ s_{i-1}<s_i $ we obtain~$ s_{i-1}{+}1<s_i $.
\qed
\end{proof}

\subsection{Combined Coins of Transfers and Donations}
\label{sect:maxnumcoins}
In this section we prove the following result, which will be applied in \cref{sect:mostcomps} to verify bounds \cref{eqn:balanceboundaugreg}, \cref{eqn:balanceboundsinglereg}, and \cref{def:edgedeb}.

\begin{lemma}
\label{lemma:kplusl}
Let~$ X $ be a path for which a \degreeone endpoint exists after creation, and assume that nodes~$ u $ and~$ v $ are matched in the creation step of~$ X $ and nodes~$ u' $ and~$ v' $ are matched in the next step.

Denote by~$ d_{u,v} $ the number of debits payed by~$ u $ and~$ v $, and by~$ d_{v'} $ the numbers of debits payed by~$ v' $ (recall that~$ u' $ pays no debits since~$ u' $ is selected with current degree one).
There exist (non-negative) integers~$ k $ and~$ l $ with
\begin{align*}
d_{u,v}\leq k ~~~\text{ and }~~~ d_{v'}\leq l 
\end{align*}
such that
\begin{align*}
k+l~\leq~2(\maxdeg-2)
\end{align*}
holds.
In particular, if~$ u' $ belongs to another component~$ Y\neq X $ and pays a donation to~$ X $, then we may choose~$ k $ as the number of donated coins.
\end{lemma}

\paragraph{Organization of the Proof.}
To verify \cref{lemma:kplusl} we distinguish the following cases.
Assume that path~$ X $ is created when~$ u $ is selected with current degree~{\boldmath$ d(u){=}2 $}.
Then we bound~$ k+l $ whether~$ u' $ is a node of~$ X $ or of another component~$ Y\neq X $, where by \cref{def:donation} the analysis of the latter case involves the use of a static donation.
Assuming that path~$ X $ is created when~$ u $ is selected with current degree~{\boldmath$ d(u)\geq3 $}, then again we bound~$ k+l $ whether~$ u' $ belongs to~$ X $ or not, where by \cref{def:donation} in the latter case we have to take into account a dynamic donation.
Cases~$ d(u)=2 $ and~$ d(u)\geq3 $ are analyzed in \cref{lemma:kplusl:2} resp. \cref{lemma:kplusltools}.

In the proof of \cref{lemma:kplusl:2} it suffices to analyze (the number of debits leaving) each of nodes~$ u,v,u' $, and~$ v' $ individually.
However, in order to show \cref{lemma:kplusltools} we have to take into account that the bound~$ l $ depends``dynamically'' on the number~$ k $ of debits payed by~$ u $ and~$ v $.
Consequently, the proof is more involved than that of \cref{lemma:kplusl:2}.
In particular, we study the endpoints neighboring nodes~$ u $ and~$ v $ as well as nodes~$ u' $ and~$ v' $ (in~\g).
Therefore we introduce some helpful notation in \cref{def:endpointsets}.

\pagebreak
\begin{lemma}
\label{lemma:kplusl:2}
\cref{lemma:kplusl} holds in case \onetwomingreedy selects~$ u $ with current degree~$ d(u)=2 $.
In particular, we have the following:
\begin{enumerate}[topsep=0mm,noitemsep,label=\alph*)]
\item\label{lemma:kplusl:2:inx}
If~$ u' $ is an~\mopt-neighbor of~$ u $ or~$ v $, i.e.\ node~$ u' $ belongs to~$ X $, then we may choose~$ k $ and~$ l $ such that we have~$ d_{u,v}=k\leq\maxdeg-2 $ as well as~$ d_{v'}=l\leq\maxdeg-2 $.
\item\label{lemma:kplusl:2:outx}
If~$ u' $ belongs to component~$ Y\neq X $, then we may choose~$ k $ and~$ l $ such that we have~$ d_{u,v}\leq k=\maxdeg-3 $ as well as~$ d_{v'}=l\leq\maxdeg-1 $.
\end{enumerate}
In both cases we have~$ k+l\leq2(\maxdeg-2) $.
\end{lemma}
\begin{proof}
To prepare the proof, consider the creation step of~$ X $.
Observe that~$ u $ is not incident with an~$ F $-edge, since we have~$ d(u)=2 $.
Hence~$ u $ pays zero debits.
Recall that also~$ u' $ pays no debits since \onetwomingreedy selects~$ u' $ with current degree~$ d'(u')=1 $.

We prove \cref{lemma:kplusl:2:inx}.
At creation of~$ X $ node~$ v $ is incident with at most~$ \maxdeg-2 $ edges of~$ F $, since~$ v $ is also incident with its~\m-edge and~\mopt-edge.
Hence~$ v $ pays at most~$ \maxdeg-2 $ debits.
We obtain~$ d_{u,v}\leq\maxdeg-2 $, since~$ u $ pays no debits.
Analogously, after creation of~$ X $ node~$ v' $ is incident with at most~$ \maxdeg-2 $ edges of~$ F $, since~$ v' $ is also a node of path~$ X $.
Here we get~$ d_{v'}\leq\maxdeg-2 $.
Now choose~$ k=d_{u,v} $ and~$ l=d_{v'} $.

We prove \cref{lemma:kplusl:2:outx}.
Since~$ u' $ does not belong to~$ X $ and is selected with degree~$ d'(u')=1 $, an~$ F $-edge~$ e $ incident with~$ u' $ is removed in the creation step of~$ X $.
Edge~$ e $ connects~$ u' $ with~$ v $, since~$ u $ is not incident with any~$ F $-edges when being selected by \onetwomingreedy.
In particular, edge~$ e $ is a static donation from~$ u' $ to~$ v $ by \cref{def:donation}.
But both~$ u' $ and~$ v $ are~\m-covered, thus~$ e $ is not a debit from~$ v $ and at most~$ \maxdeg-3 $ debits leave~$ v $.
We get~$ d_{u,v}\leq\maxdeg-3 $.
Now observe that the~\m-edge incident with~$ v' $ is not a debit from~$ v' $.
Hence we have~$ d_{v'}\leq\maxdeg-1 $.
Now choose~$ k $ as the number of donated coins~$ k=\maxdeg-3 $ as well as~$ l=d_{v'} $.
\qed
\end{proof}

\begin{definition}
\label{def:endpointsets}
Consider the creation step of path~$ X $, when \onetwomingreedy selects node~$ u $ with current degree~$ d(u)\geq3 $ and matches~$ u $ with~$ v $.
We denote
\begin{itemize}[topsep=0mm,noitemsep]
\item 
by~$ \mathcal{W} $ the set of path endpoints~$ w $ with current degree at least~$ d(w)\geq3 $,
\item
by~$ W\subseteq\mathcal{W} $ the set of~$ \mathcal{W} $-endpoints neighboring~$ u $ or~$ v $,
\item 
by~$ W_\delta=\{w\in W:d'(w)=\delta\} $ the set of~$ W $-endpoints which have degree~$ \delta $ after creation of~$ X $ for~$ \delta\in\{1,2\} $,
\item 
by~$ W_{\geq3}=W\setminus(W_1\cup W_2) $ the set of~$ W $-endpoints which have degree at least~$ 3 $ after creation of~$ X $,
\item 
by~$ W_1^f=\{w\in W_1:|\{\edge{w,u},\edge{w,v}\}\cap F|=f\,\} $ the set of~$ W_1 $-endpoints connected to~$ u $ and~$ v $ by~$ f $~edges of~$ F $ for~$ f\in\{1,2\} $, and
\item
by~$ \mathcal{E}(W)=\{\,\{x,y\}\in E:x\in\{u,v\},y\in W\,\} $ the set of edges connecting nodes~$ u $ and~$ v $ with~$ W $-endpoints.
\end{itemize}
Note that~$ W_1^1 $ and~$ W_1^2 $ form a partition of~$ W_1 $ and that~$ W_1 $,~$ W_2 $, and~$ W_{\geq3} $ form a partition of~$ W $, i.e.\ sets~$ \mathcal{W}\setminus W, W_{\geq3},W_2,W_1^2 $, and~$ W_1^1 $ are pairwise disjoint.
\end{definition}

\begin{lemma}
\label{lemma:kplusltools}

\cref{lemma:kplusl} holds in case \onetwomingreedy selects~$ u $ with current degree~$ d(u)\geq3 $.
In particular, we may choose~$ k $ and~$ l $ such that we have
\begin{align}
d_{u,v}~=~k&~=~2|W_1^2|+|W_1^1|~~~~\text{ and}\tag{a}\label{lemma:kplusltools:k}\\
d_{v'}~=~l\;&~\leq~|W_1^1|+|W_2|\tag{b}\label{lemma:kplusltools:l}
\end{align}
as well as
\begin{align}
k+l&~\leq~2|W_1^2|+2|W_1^1|+|W_2|\notag\\
&~\leq~|\mathcal{E}(W)|\tag{c}\label{lemma:kplusltools:e}\\
&~\leq~2(\maxdeg-2)\tag{d}\label{lemma:kplusltools:bound}\,.
\end{align}
\end{lemma}
\begin{proof}
We prove  \cref{lemma:kplusltools:k}.
Therefore we show that nodes~$ u $ and~$ v $ pay~$ d_{u,v}=2|W_1^2|+|W_1^1| $ debits.
To prove that we may choose~$ k=d_{u,v} $ we have to verify that if~$ X $ receives a donation then this donation moves~$ d_{u,v} $ coins.
By \cref{def:donation}, path~$ X $ receives a dynamic donation since~$ u $ was selected with current degree~$ d(u)\geq3 $.
Now observe that a dynamic donation moves exactly~$ d_{u,v} $ coins.

To count the number of debits payed by~$ u $ and~$ v $, we first study which endpoints do \emph{not} receive a transfer from~$ u $ or~$ v $.
First, recall that each endpoint~$ w\in\mathcal{W} $ has degree at least three before creation of~$ X $ and observe that~$ w $ has degree at least one thereafter, since degrees drop by at most two when~$ X $ is created.
Endpoints in~$ \mathcal{W}\setminus W $, i.e.\ endpoints which are not adjacent to~$ u $ or~$ v $, as well as endpoints in~$ W_{\geq3} $ still have degree at least three after creation of~$ X $ and hence do not receive credits from~$ u $ or~$ v $.
Also, endpoints in~$ W_2 $ receive no credits from~$ u $ or~$ v $, since their degree is at least two after creation.

Hence only endpoints in~$ W_1 $ might receive credits from~$ u $ or~$ v $, since their degree is exactly one after creation.
In the creation step of~$ X $, for each~$ w\in W_1 $, either two incident~$ F $-edges are removed, i.e.\ we have~$ w\in W_1^2 $, or one incident~$ F $-edge and an incident~\mopt-edge is removed, i.e.\ we have~$ w\in W_1^1 $:
since all these removed~$ F $-edges are transfers, we get that nodes~$ u $ and~$ v $ pay exactly~$ d_{u,v}=2|W_1^2|+|W_1^1| $ debits.

To show \cref{lemma:kplusltools:l}, we prove that node~$ v' $ pays at most~$ d_{v'}\leq|W_1^1|+|W_2| $ debits and set~$ l=d_{v'} $.
Which endpoints do \emph{not} receive transfers from~$ v' $?
(Recall that~$ u' $ pays no debits since~$ u' $ is selected when it has degree~$ d'(u')=1 $.)
Here, our \cref{def:cancel} of transfer cancellations is crucial:
observe that~$ v' $ does not pay transfers to endpoints in~$ W_1^2 $, since each of these endpoints already receives two credits from~$ u $ and~$ v $ and any additional credit from~$ v' $ is canceled by \cref{def:cancel}.
Furthermore, endpoints in~$ \mathcal{W}\setminus W $ have degree at least three after creation of~$ X $ and do not receive transfers from~$ v' $:
when~$ u' $ and~$ v' $ are matched, node~$ u' $ is selected with current degree~$ d'(u')=1 $, thus the degree of each endpoint in~$ \mathcal{W}\setminus W $ drops by at most one and to at least two.
Analogously, no endpoint in~$ W_{\geq3} $ receives a transfer from~$ v' $.

Hence~$ v' $ might pay debits only to endpoints in~$ W_1^1 $ and~$ W_2 $, and \cref{lemma:kplusltools:l} follows.
We note that, in particular, an endpoint~$ w\in W_1^1\cup W_2 $ receives \emph{less} than two credits from~$ u $ or~$ v $, i.e.\ further credits are not canceled, and we have~$ d'(w)\leq2 $, i.e.\ the degree of~$ w $ might drop to at most one when edge~\edge{v',w} is removed.

We prove \cref{lemma:kplusltools:e}, i.e.\ that~$ k+l $ is bounded from above by the number~$ |\mathcal{E}(W)| $ of edges connecting~$ u $ or~$ v $ with an endpoint before creation of~$ X $.
First, we argue that 
$$ |W_2|\leq |\mathcal{E}(W)|-2(|W_1^1|+|W_1^2|) $$
\pagebreak

\noindent
holds:
for each~$ w\in W_1^1 $ as well as for each~$ w\in W_1^2 $ two edges connecting~$ w $ with~$ u $ and~$ v $ are removed when in the creation step of~$ X $ the degree of~$ w $ drops from at least~$ d(w)\geq 3 $ to at most~$ d'(w)\leq1 $;
moreover, each endpoint in~$ W_2 $ is connected with~$ u $ or~$ v $ by at least one of the remaining edges.
Now we apply \cref{lemma:kplusltools:k,lemma:kplusltools:l} to obtain~$ k+l\leq2|W_1^2|+2|W_1^1|+|W_2|\leq|\mathcal{E}(W)| $, as claimed.

We prove \cref{lemma:kplusltools:bound}, i.e.\ that~$ |\mathcal{E}(W)| $ is bound from above by~$ 2(\maxdeg-2) $.
To see this, observe first that before creation of~$ X $ two edges incident with~$ u $ do not connect~$ u $ with an endpoint:
edge~\edge{u,v} is an~\m-edge and edge~\edge{u,u'} is removed when in the creation step of~$ X $ the degree of~$ u' $ drops from at least~$ d(u')\geq d(u)\geq3 $ to~$ d'(u')=1 $.
Analogously, edges~\edge{v,u} and~\edge{v,u'} do not connect~$ v $ with an endpoint.
Consequently, at most~$ \maxdeg-2 $ edges connect each of~$ u $ and~$ v $ with an endpoint.
The statement follows.
\qed
\end{proof}

\subsection{All Components Are Balanced}
\label{sect:mostcomps}
Recall that in order to prove \cref{thm:mingreedy4} we have to verify \cref{eqn:mintwocreds}, \cref{def:edgedeb},\linebreak
\cref{eqn:balanceboundaugreg}, and \cref{eqn:balanceboundsinglereg}.
Recall also that in \cref{prop:newminmaxcred} we have already shown\linebreak \cref{eqn:mintwocreds}, i.e.\ that each path receives at least two credits at its endpoints.

\smallskip

In \cref{lemma:maugandsingle} we show \cref{eqn:balanceboundsinglereg} for each singleton.

\smallskip

Also in \cref{lemma:maugandsingle}, we show \cref{def:edgedeb} for path~\m-edges.
Thereafter it remains to prove \cref{eqn:balanceboundaugreg} for each path.
In particular, we prove \cref{eqn:balanceboundaugreg} for each path for which a \degreeone endpoint exists after creation as well as for each path for which there does not exist a \degreeone endpoint after creation:
these proofs are provided in \cref{lemma:nonbip4augdeg1} resp. \cref{lemma:nonbip4exceptional}.

In the proof of \cref{eqn:balanceboundaugreg} we may assume that the nodes~$ u $ and~$ v $ matched to create a path~$ X $ pay at least one debit, since otherwise at least~$ 2(\maxdeg-2) $ debits are missing for~$ X $ and as required in \cref{eqn:balanceboundaugreg} the balance of~$ X $ is at least~$ c_X-d_X\geq2-D_X+2(\maxdeg-2) $, by \cref{eqn:mintwocreds} and \cref{def:edgedeb}.

\bigskip

We begin with the proof of \cref{def:edgedeb,eqn:balanceboundsinglereg}.

\begin{lemma}
\label{lemma:maugandsingle}
Bounds \cref{def:edgedeb,eqn:balanceboundsinglereg} hold, since for each~\m-edge~\edge{y,y'} of a path or a singleton, nodes~$ y $ and~$ y' $ pay at most~$ 2(\maxdeg-2) $ coins.
\end{lemma}
\begin{proof}
We have to verify that~$ y $ and~$ y' $ pay at most~$ 2(\maxdeg-2) $ coins whether a donation must be payed or not.
Assume that \onetwomingreedy selects node~$ y $ and matches~$ y $ with neighbor~$ y' $.
Hence a donation might be payed from~$ y $, but not from~$ y' $.

\medskip

Assume that~$ y $ pays a donation.
The following argument applies whether~\edge{y,y'} is a path~\m-edge or a singleton.
By \cref{def:donation} of donations, edge~\edge{y,y'} is picked in the donation step of a path~$ X $ for which a \degreeone endpoint exists after creation.
By \cref{lemma:kplusl}, we have~$ k+l\leq2(\maxdeg-2) $ for the number~$ k $ of coins donated by~$ y $ and the number~$ l $ of debits payed by~$ y' $, i.e.\ both nodes pay at most~$ 2(\maxdeg-2) $ coins.

\medskip

From here on we assume that nodes~$ y $ and~$ y' $ pay only debits.
If~\edge{y,y'} is a path~\m-edge then observe that each of~$ y $ and~$ y' $ is incident with at most~$ \maxdeg-2 $ edges of~$ F $ and recall that only~$ F $-edges can be transfers.
Thus~$ y $ and~$ y' $ pay at most~$ 2(\maxdeg-2) $ coins.

\smallskip

If~\edge{y,y'} is the edge of a singleton~$ Y $, then each of~$ y $ and~$ y' $ has at most~$ \maxdeg-1 $ incident~$ F $-edges, i.e.\ we have~$ d_Y\leq2(\maxdeg-1) $.
In order to prove \cref{eqn:balanceboundsinglereg} with balance at least~$ -d_Y\geq-2(\maxdeg-2) $ it suffices to show that at least two debits are missing.
Assume otherwise, i.e.\ that~$ y $ and~$ y' $ pay at least~$ 2(\maxdeg-2)+1 $ debits.
Then each of~$ y $ and~$ y' $ has degree at least~$ \maxdeg-1 $ when~\edge{y,y'} is picked, since otherwise one of~$ y $ and~$ y' $ would have at most~$ \maxdeg-3 $ incident~$ F $-edges and hence at least two missing debits.
But since both~$ y $ and~$ y' $ have degree at least~$ \maxdeg-1\geq3 $ before creation and a debit leaves~$ y $ or~$ y' $, say to endpoint~$ w $, before creation endpoint~$ w $ has degree at least~$ d(w)\geq3 $ as well and we obtain that~$ d'(w)=1 $ holds after creation.
Hence a node other than~$ w $ is selected next after creation of~$ Y $, call it~$ z $.
But~$ z $ had degree at least~3 before creation of~$ Y $ as well, hence~$ z $ must be a \degreeone node after creation:
two edges connecting~$ y $ and~$ y' $ with~$ z $ are removed from the graph in the creation step of~$ Y $.
Both edges are not transfers, since~$ y,y' $, and~$ z $ are~\m-covered.
Consequently, each of~$ y $ and~$ y' $ has a missing debit, a contradiction.
\qed
\end{proof}

To complete the proof of \cref{thm:mingreedy4}, it remains to verify \cref{eqn:balanceboundaugreg}.

\begin{lemma}
\label{lemma:nonbip4augdeg1}
Bound \cref{eqn:balanceboundaugreg} holds for each path for which a \degreeone path endpoint exists after creation.
\end{lemma}
\begin{proof}
Let~$ X $ be a path for which a \degreeone endpoint exists after creation.
Since~$ X $ receives at least~$ c_X\geq2 $ credits by \cref{eqn:mintwocreds} and each~\m-edge of~$ X $ pays at most~$ 2(\maxdeg-2) $ coins by \cref{def:edgedeb}, to show \cref{eqn:balanceboundaugreg} we have to prove that the balance of~$ X $ is increased above~$ 2-D_X $ by at least~$ 2(\maxdeg-2) $.

Let~$ u' $ be the node selected next after creation of~$ X $, and recall that~$ u' $ is selected with degree~$ d'(u)=1 $ since there is a \degreeone endpoint after creation of~$ X $.
By \cref{lemma:u'props}, node~$ u' $ is the~\mopt-neighbor of~$ u $ or~$ v $, or~$ u' $ belongs to another component~$ Y\neq X $.
We distinguish these two cases.

Assume that~$ u' $ is the~\mopt-neighbor of~$ u $ or~$ v $, i.e.\ edge~\edge{u',v'} is an~\m-edge of~$ X $.
Observe that by \cref{def:edgedeb} the number of coins payed by nodes of the~$ X $-edges~\edge{u,v} and~\edge{u',v'} is bounded from above by~$ 4(\maxdeg-2) $.
However, letting~$ k $ be the number of debits payed by nodes~$ u $ and~$ v $, and letting~$ l $ be the number of debits payed by nodes~$ u' $ and~$ v' $, \cref{lemma:kplusl} shows that~$ k+l\leq2(\maxdeg-2) $ holds, i.e.\ edges~\edge{u,v} and~\edge{u',v'} pay at least~$ 2(\maxdeg-2) $ coins less than maximum.
Therefore we obtain~$ d_X\leq D_X-2(\maxdeg-2) $ and hence the balance of~$ X $ is at least~$ c_X-d_X\geq2-D_X+2(\maxdeg-2) $.

Now assume that~$ u' $ belongs to another component~$ Y\neq X $.
Hence~$ X $ receives a donation from~$ u' $.
Here, \cref{lemma:kplusl} shows that the number~$ k $ of coins received by~$ X $ over the donation satisfies~$ k\geq d_{u,v} $ for the number~$ d_{u,v} $ of debits payed by~$ u $ and~$ v $.
Since we have~$ d_{u,v}-k\leq0 $, we obtain~$ d_X\leq(\w{X}-1)\cdot2(\maxdeg-2)+d_{u,v}{-}k\leq (\w{X}-1)\cdot2(\maxdeg-2)=D_X-2(\maxdeg-2) $ and hence the balance of~$ X $ is at least~$ c_X-d_X\geq2-D_X+2(\maxdeg-2) $ again.
\qed
\end{proof}

\paragraph{No Degree-1 Endpoint Exists After Creation.}
\newcommand{\barX}{\ensuremath{{\bm\bar{X}}}\xspace}
\newcommand{\barx}{\ensuremath{{\bm\bar{x}}}\xspace}
\newcommand{\baru}{\ensuremath{{\bm\bar{u}}}\xspace}
\newcommand{\barv}{\ensuremath{{\bm\bar{v}}}\xspace}
\newcommand{\bard}{\ensuremath{{\bm\bar{d}}}\xspace}
\newcommand{\barw}{\ensuremath{{\bm\bar{w}}}\xspace}
By~$ \barX $ we denote path such that after creation of~$ \barX $ there does not exist a \degreeone endpoint.
Assume that \onetwomingreedy selects node~$ \baru $ to create~$ \barX $ and matches~$ \baru $ with~$ \barv $.
\cref{lemma:createpathdeg1}\cref{lemma:createpathdeg1:nodebu} shows that node~$ \baru $ pays no debits, \cref{lemma:createpathdeg1}\cref{lemma:createpathdeg1:debv,lemma:createpathdeg1:onedeb} show that node~$ \barv $ pays exactly one debit.
Therefore we can only bound the number of missing debits from~$ \baru $ and \barv by at least~$ 2(\maxdeg-2)-1 $.
Thus the number of coins payed by~$ {\barX} $ is at most
$$ d_{{\barX}}\leq D_{{\barX}}-(\,2(\maxdeg-2)-1\,)\,, $$
where~$ D_\barX=\w{\barX}\cdot2(\maxdeg-2) $ is the maximum possible number of coins payed by \barX, since each~\m-edge of~\barX pays at most~$ 2(\maxdeg-2) $ coins by \cref{def:edgedeb}.

Can we find an additional missing debit
However, unlike in the case that a \degreeone endpoint exists after creation, in order to find an additional missing debit for~$ {\barX} $ we cannot rely on the analysis of a \degreeone node matched after creation of~\barX.
In particular, no coins are donated to~\barX.

Therefore, using \cref{eqn:mintwocreds} we obtain that the balance of~\barX is at least
$$c_{{\barX}}-d_{{\barX}}\geq2-(D_{{\barX}}-(\,2(\maxdeg-2)-1\,))=1-D_{{\barX}}+2(\maxdeg-2)\,,$$
which fails to satisfy bound \cref{eqn:balanceboundaugreg} by one coin.
Thus, if we can find an additional credit or an additional missing debit, then~$ {\barX} $ is balanced and we are done.
Hence the following result completes the proof of \cref{thm:mingreedy4}.

\begin{lemma}
\label{lemma:nonbip4exceptional}
Bound \cref{eqn:balanceboundaugreg} holds for each path~\barX for which there does not exist a \degreeone path endpoint after creation, since~$ {\barX} $ receives~$ c_{{\barX}}\geq3 $ credits or an~\m-edge~\edge{{\barx},{\barx}'} of~$ {\barX} $ pays at most~$ 2(\maxdeg-2)-1 $ coins, where~$ \edge{{\barx},{\barx}'}\neq\edge{\baru,\barv} $ is not the~\m-edge of~\barX picked in the creation step.
\end{lemma}
\begin{proof}
To prepare the proof, observe that to show~$ c_{{\barX}}\geq3 $ it suffices to identify an endpoint of~$ {\barX} $ which receives at least two credits, since by \cref{prop:numcreditsreg} the other endpoint of~$ {\barX} $ receives an additional credit.

We distinguish cases by the number~$ \w{\barX} $ of~\m-edges of~\barX.

\medskip

Assume that~$ {\barX} $ is a~\oneOverTwo-path with~{\boldmath$ \w{{\barX}}=1 $}.
We show that an endpoint of~$ {\barX} $ receives two credits, which proves~$ c_{{\barX}}\geq3 $.
First, we consider the creation step of~\barX and argue that in the subsequent step there is an endpoint of~$ {\barX} $ which is not isolated.
Since after creation of~\barX there does not exist a \degreeone endpoint, by \cref{lemma:createpathdeg1}\cref{lemma:createpathdeg1:deg2} the \onetwomingreedy algorithm selects a node~\baru of current degree exactly~$ \bard(\baru)=2 $.
Let~$ \barv $ denote the node matched with~$ \baru $ and observe that the~\mopt-neighbor of~$ \barv $, call it~$ \barw, $ has degree at least~$ \bard(\barw)\geq\bard(\baru)=2 $ as well.
Since~$ \w{\barX}=1 $ holds, node~$ \barw $ must be an endpoint of~$ {\barX} $, and since we have~$ \bard(\baru)=2 $ endpoint~$ \barw $ is not adjacent to~$ \baru $.
Thus the degree of~$ \barw $ drops by at most one and to at least~$ \bard'(\barw)\geq1 $ in the step after creation of~$ {\barX} $.

Moreover, after creation of~$ {\barX} $ all non-isolated path endpoints have degree at least two, since there does not exist a \degreeone endpoint.
This holds in particular for~$ \barw $.
But since after creation of~$ {\barX} $ endpoint~$ \barw $ is incident with at least two~$ F $-edges, endpoint~$ \barw $ eventually receives at least two credits, namely over those two~$ F $-edges of~$ \barw $ which are removed last from the graph.

\medskip

From here on, assume that~$ {\barX} $ has~{\boldmath$ \w{{\barX}}\geq2 $} edges of~\m.
We again denote the nodes matched to create~$ {\barX} $ by~$ \baru $ and~$ \barv $, where we assume that \onetwomingreedy selects~$ \baru $.
Recall that since no \degreeone endpoint exists after creation of~$ {\barX} $, by \cref{lemma:createpathdeg1}\cref{lemma:createpathdeg1:deg2} an~\mopt-edge of~$ {\barX} $ incident with an endpoint of~$ {\barX} $ is removed from the graph.

We consider the step~$ s $ when for the second time an~\mopt-edge incident with an endpoint of~$ {\barX} $ is removed from the graph.
Denote by~$ {\barx} $ the selected node and by~$ {\barx}' $ the neighbor matched with~$ {\barx} $.
Since~$ \w{{\barx}}\geq2 $ holds we have~$ \edge{\baru,\barv}\neq\edge{{\barx},{\barx}'} $.
We distinguish the cases that~$ {\barx} $ and~$ {\barx}' $ pay only debits, or~$ {\barx} $ pays a static or a dynamic donation.

\smallskip

First, assume that~$ {\barx} $ and~$ {\barx}' $ pay {\bfseries no donation} but only debits.
To show that edge~\edge{{\barx},{\barx}'} pays at most~$ 2(\maxdeg-2)-1 $ debits, assume the opposite.
Then each of~$ {\barx} $ and~$ {\barx}' $ pays~$ \maxdeg-2 $ debits, since each is incident with at most~$ \maxdeg-2 $ edges of~$ F $.
Hence each of~$ {\barx} $ and~$ {\barx}' $ has degree at least~$ \maxdeg-1 $ at step~$ s $, since each is incident with~$ \maxdeg-2 $ edges of~$ F $ as well as with edge~\edge{{\barx},{\barx}'}.
But a debit leaves~$ {\barx} $ or~$ {\barx}' $, say to endpoint~$ w $, hence by \cref{def:transfer} of transfers after step~$ s $ endpoint~$ w $ has degree at most one.
In particular, endpoint~$ w $ has degree exactly one after step~$ s $, since before step~$ s $ node~$ {\barx} $ has degree at least~$ \maxdeg-1\geq3 $ and hence endpoint~$ w $ has degree at least~$ 3 $ as well.

Thus a node other than~$ w $ is selected next, call it~$ y $.
But node~$ y $ had degree at least~$ 3 $ before step~$ s $ as well and since thereafter endpoint~$ w $ has degree exactly one, node~$ y $ is selected with current degree exactly one.
Consequently, edges~\edge{\barx,y} and~\edge{\barx',y} are removed from the graph in step~$ s $.
Both~\edge{\barx,y} and~\edge{\barx',y} are not transfers, since~$ {\barx},{\barx}' $, and~$ y $ are~\m-covered.
Now observe that at most one of~\edge{\barx,y} and~\edge{\barx',y} can be an~\mopt-edge, since otherwise two~\mopt-edges would be incident with~$ y $:
at least one of~\edge{\barx,y} and~\edge{\barx',y} is an~$ F $-edge, say~\edge{\barx,y}.
Since edge~\edge{\barx,y} is not a transfer, there is a debit missing for~$ {\barx} $ and edge~\edge{{\barx},{\barx}'} pays at most~$ 2(\maxdeg-2)-1 $ debits, as claimed.

\smallskip

Now assume that~\barx pays a {\bfseries static donation}.
By \cref{def:donation}, path~$ {\barX} $ pays~$ \maxdeg-3 $ coins over the donation from~$ {\barx} $.
Furthermore, node~$ {\barx}' $ pays at most~$ \maxdeg-2 $ debits, since at most~$ \maxdeg-2 $ many~$ F $-edges are incident with~$ {\barx}' $ when \onetwomingreedy picks edge~\edge{{\barx},{\barx}'}.
Again, the nodes in edge~\edge{{\barx},{\barx}'} pay at most~$ 2(\maxdeg-2)-1 $ coins.

\smallskip

Lastly, assume that~$ {\barx} $ pays a {\bfseries dynamic donation}~\don{{\barx},u} of~$ k $ coins.
By \cref{def:donation} of donations, node~$ u $ belongs to a path~$ X\neq\barX $ for which a \degreeone endpoint exists after creation, and \onetwomingreedy selects~$ u $ with current degree~$ d(u)\geq3 $ when creating~$ X $.
Edge~\edge{\barx,\barx'} is picked in the donation step of~$ X $ when~\barx has current degree~$ d'(\barx)=1 $.

Assume that~$ u $ is matched with~$ v $, and recall the definition of the sets~$ W,W_1$, $W_1^1,W_1^2,W_2 $, and~$ W_{\geq3} $ of endpoints neighboring~$ u $ and~$ v $, see \cref{def:endpointsets}.
By \cref{lemma:kplusl,lemma:kplusltools} we have the following:
nodes~$ u $ and~$ v $ pay~$ d_{u,v}=k=2|W_1^2|+|W_1^1| $ debits, node~\barx donates~$ k $ coins to~$ u $, node~$ {\barx}' $ pays at most~$ d_{{\barx}'}=l\leq|W_1^1|+|W_2| $ debits.
Hence~\barx and~$ \barx' $ pay at most~$ k+l\leq2(\maxdeg-2) $ coins in total.
If~$ k+l<2(\maxdeg-2) $ holds then an additional debit is missing from~\barx or~$ \barx' $ and we are done.

\smallskip

So assume from here on that~\barx and~$ \barx' $ pay~$ k+l=2(\maxdeg-2) $ coins.
As a consequence of \cref{lemma:kplusltools}\cref{lemma:kplusltools:e,lemma:kplusltools:bound} we get
$$ k+l=2|W_1^2|+2|W_1^1|+|W_2|=|\mathcal{E}(W)|=2(\maxdeg-2)\,. $$
Since~$ k=2|W_1^2|+|W_1^1| $ holds, we get that~$ \barx' $ pays exactly
$$ d_{{\barx}'}=l=|W_1^1|+|W_2| $$
debits.
To which endpoints?
Node~$ \barx' $ pays zero debits to endpoints in~$ \mathcal{W}\setminus W $ or~$ W_{\geq3} $:
in the step when \onetwomingreedy picks edge~\edge{\barx,\barx'} these endpoints have degree at least three, and since these endpoints are not adjacent to~\barx when~\barx is selected with current degree~$ d'(\barx)=1 $ in the donation step of~$ X $, the degrees of these endpoints drop by at most one and to at least two, i.e.\ these endpoints do not receive transfers from~$ \barx' $.
Therefore~$ \barx' $ might only pay debits to endpoints in~$ W_2\cup W_1 $.
But node~$ \barx' $ also pays no debits to nodes in~$ W_1^2 $, since in the step when \onetwomingreedy picks edge~\edge{\barx,\barx'} each endpoint in~$ W_1^2 $ already receives two credits from nodes~$ u $ and~$ v $ and any further credits from~$ \barx' $ are canceled.
Consequently, node~$ {\barx}' $ might pay debits only to nodes in~$ W_2 $ or~$ W_1^1 $.
Since sets~$ \mathcal{W}\setminus W,W_{\geq3},W_2,W_1^2 $, and~$ W_1^1 $ are pairwise disjoint and since we have~$ d_{\barx'}=|W_1^1|+|W_2| $, we obtain that~$ \barx' $ pays exactly one debit to each endpoint in~$ W_1^1 $ as well as to each endpoint in~$ W_2 $.

\medskip

In the rest of the proof we proceed as follows.
We show that either an endpoint of~\barX receives at least two credits, which proves~$ c_X\geq3 $, or we show a contradiction to our assumption that nodes~\barx and~$ \barx' $ pay~$ k+l=2(\maxdeg-2) $ coins, which proves that edge~\edge{\barx,\barx'} pays at most~$ 2(\maxdeg-2) $ coins.

Consider the step when \onetwomingreedy picks edge~\edge{\barx,\barx'}.
Recall that~\barx is selected with current degree~$ d'(\barx)=1 $, i.e.\ when the~\mopt-edge of~\barx is already removed from the graph.
Recall also that an~\mopt-edge incident with a path endpoint of~\barX is removed in this step.
This~\mopt-edge must be incident with~$ \barx' $ and we call it edge~\edge{\barx',\barw}.

We conduct a case analysis based on which of~$ W_1^1,W_1^2 $, or~$ W_2 $ endpoint~$ \barw $ belongs to.
(Recall that sets~$ W_1^1$ and~$ W_1^2$ form a partition of~$ W_1 $ and sets~$ W_1, W_2$, and~$W_{\geq3}$ form a partition of~$W $.)
Path~$ X $ is created when the nodes~$ u $ and~$ v $ being matched have degree at least~$ d(u)\geq3 $ resp.~$ d(v)\geq3 $, and after creation of~$ X $ the degree of~$ {\barx} $ is exactly~$ d'({\barx})=1 $.
Thus two edges~\edge{u,\barx} and~\edge{v,\barx} are removed from the graph in the creation step of~$ X $.

\begin{itemize}[leftmargin=2cm]
{
\setlength\itemindent{3.2cm}
\item[$ \barw\in\mathcal{W}\setminus W $ or~$ \barw\in W_{\geq3} $~:]
Here we show that~\barw receives at least two credits due to its large degree after creation of~$ X $.

Since~$ X $ receives a dynamic donation, by \cref{def:donation} of donations the degrees of~$ u $ and~$ v $ are at least three in the creation step of~$ X $.
Consequently, the degree of~\barw is at least~$ d(\barw)\geq3 $ as well.
If we have~$ \barw\in\mathcal{W}\setminus W $, then~\barw is not adjacent to~$ u $ or~$ v $ and hence no edges incident with~\barw are removed in the creation step of~$ X $.
Thus the degree of~\barw is at least~$ d'(w)\geq3 $ after creation of~$ X $.
If~$ \barw\in W_{\geq3} $ holds, then the degree of~\barw is at least~$ d'(w)\geq3 $ after creation of~$ X $ by \cref{def:endpointsets}.

In the step after creation of~$ X $, i.e.\ in the donation step of~$ X $, when \onetwomingreedy picks edge~\edge{\barx,\barx'}, the degree of endpoint~\barw drops by at most one, since~\barw is not adjacent with node~\barx which is selected with current degree~$ d'(\barx)=1 $.

Consider the step after edge~\edge{\barx,\barx'} is picked.
The degree of~$ \barw $ is at least two and the~\mopt-edge of~$ \barw $ is removed from the graph.
Hence~$ \barw $ is incident with at least two~$ F $-edges.
Now observe that~$ \barw $ eventually receives at least two credits, namely over the two of the~$ F $-edges of~$ \barw $ which are removed last from the graph.
Consequently, path~$ {\barX} $ receives at least~$ c_{{\barX}}\geq3 $ credits.
}
\item[$ \barw\in W_1^1 $~:]
We show a contradiction to our assumption that~$ k+l=2(\maxdeg-2) $ holds.

Consider the step when path~$ X $ is created.
Recall by \cref{def:endpointsets} of~$ W_1^1 $, that at least two edges incident with~$ \barw $ are removed, since~$ \barw $ has degree at least~$ d(\barw)\geq3 $ before and degree at most~$ d'(\barw)\leq1 $ afterwards.
Also by definition of~$ W_1^1 $, only one~$ F $-edge incident with~$ \barw $ is removed, hence the~\mopt-edge incident with~$ \barw $ is removed as well.
Consequently, endpoint~$ \barw $ belongs to~$ X $.

We obtain a contradiction, since endpoint~$ \barw $ belongs to path~$ {\barX} $ and~$ {\barX}{\neq} X $ holds by definition of the donation~\don{{\barx},u}.
\item[$ \barw\in W_1^2 $~:]
By \cref{def:endpointsets} of set~$ W_1^2 $, two~$ F $-edges~$ e_1 $ and~$ e_2 $ incident with~$ \barw $ are removed when path~$ X $ is created.
In particular, since after creation of~$ X $ the degree of~$ \barw $ is exactly one, edges~$ e_1 $ and~$ e_2 $ are credits to~\barw and both credits are never canceled.
Hence path~$ {\barX} $ receives at least~$ c_{{\barX}}\geq3 $ credits.
\item[$ \barw\in W_2 $~:]
Recall that~$ d_{{\barx}'}=l=|W_1^1|+|W_2| $ holds and that node~$ {\barx}' $ pays a debit to each endpoint in~$ W_1^1 $ and to each endpoint in~$ W_2 $.
Since we have~$ \barw\in W_2 $, node~$ {\barx}' $ pays a debit to~$ \barw $.
But only~$ F $-edges can be transfers, thus we have~$ \edge{{\barx}',\barw}\in F $.
A contradiction to endpoint~$ \barw $ being the~\mopt-neighbor of~$ {\barx}' $, i.e.\ to~$ \edge{{\barx}',\barw}\in\mopt $.
\qed
\end{itemize}
\end{proof}

\section{Inapproximability Results}
\label{hardinstances}

\mingreedy does not exploit knowledge gathered about the input in previous steps:
e.g. the neighbors of the selected node $ u $ are not remembered in order to ``explore'' the neighborhood of $ u $ later.
In a step of \mingreedy, an arbitrary node of minimum degree, who is located in an unknown place in the graph, is matched to an arbitrary neighbor.

A question arises naturally:
are the worst case performance guarantees given above for \mingreedy optimal, i.e.\ is there a greedy matching algorithm which always computes larger matchings than proven for \mingreedy?
In particular, a greedy matching algorithm in question may in each step utilize all previously gathered knowledge in very sophisticated node and edge selection routines.

\emph{Adaptive priority algorithms} \cite{DBLP:conf/soda/BorodinNR02} in the \emph{vertex model} \cite{DBLP:conf/soda/DavisI04} define a large class of deterministic greedy matching algorithms, which we denote as \greedymatching.
\greedymatching-algorithms  do not have resource constraints and formalize the essential properties of greedy algorithms:
to what extent can the input be unveiled in a single step, what are the possible irrevocable decisions for the constructed solution to be done after part of the input is revealed?
An \greedymatching-algorithm may gather and process much data about its input instances and deduce knowledge to be used in clever future steps.
Therefore \greedymatching-algorithms seem much stronger than \mingreedy.
In particular,~\greedymatching contains (the deterministic variants of) many prominent greedy matching algorithms, see \cref{lemma:apvcontainswhat}.

Nevertheless, we construct graphs with degrees bounded by \maxdeg for which a matching of size at most $ \desiredRatio+o(1) $ times optimal is computed.
So our \twoOverThree lower bound for $\maxdeg=3$ and our \desiredRatio lower bound for \maxdeg-regular graphs are tight:
the very simple \mingreedy algorithm has optimal worst case performance among \greedymatching-algorithms.
For graphs of degree at most \maxdeg our \weakerRatio lower bound shows that \mingreedy has good worst case performance.

For an \greedymatching algorithm $ A $ the input graph is represented as a set of adjacency lists, e.g. $ \{\langle u;v,w\rangle,\langle v;u,w\rangle, \langle w;u,v\rangle\} $ is the triangle on $ u,v,w $ (where an arbitrarily ordered list of neighbors appears after a semicolon).
An~\greedymatching-algorithm~$ A $ has a priori access to the number of nodes and starts with an empty matching.
In each step, algorithm $ A $ selects a node by specifying a total priority order on all possible adjacency lists.
From the given order, algorithm $ A $ receives the adjacency list which has highest priority and corresponds to a still non-isolated node $ u $, say $ \langle u;v,w,\dots\rangle $.
(A node is called \emph{isolated}, if it, or each of its neighbors, is matched.)
Lastly, algorithm $ A $ selects a non-isolated matching partner for $ u $ from the neighbor set $ \{v,w,\dots\}$ and then changes to the next step.

Matched nodes are not removed from the adjacency lists in~$ G $:
\cref{obs:eitherMatchedOrIsolated} shows that if~$ A $ remembers already matched nodes, then~$ A $ can submit priority orders on adjacency lists w.r.t. the ``reduced'' graph, i.e.\ w.r.t. the set of nodes which are not yet isolated.

\begin{observation}
\label{obs:eitherMatchedOrIsolated}
If an \greedymatching-algorithm receives a data item~$ \langle u;v,w,x,\dots\rangle $, then each neighbor~$ v,w,x,\dots $ of~$ u $ is either matched or not isolated.
\end{observation}
\begin{proof}
Let~$ y \in\{v,w,x,\dots\}$ and assume that~$ y $ is not matched but isolated.
Each neighbor of~$ y $ is matched and thus isolated.
Since~$ u $ is a neighbor of~$ y $ and the data item of~$ u $ is received, node~$ u $ is not isolated.
A contradiction.
\qed
\end{proof}

Leaving adjacency lists unchanged increases the power of $ A $:
a neighbor of an already matched node $ v $ may be requested and hence $ A $ is able to explore the neighborhood of $ v $ and is not oblivious to the parts of \g being processed.

\pagebreak
\begin{lemma}
\label{lemma:apvcontainswhat}
The class \greedymatching contains (all deterministic variants of) \greedy, \karpsipser, \mrg, \mingreedy, \shuffle, and all vertex iterative algorithms.
\end{lemma}
\begin{proof}
We have to implement any deterministic variant of one of the given algorithms as an~\greedymatching-algorithm.

In a given round of \greedy, a priority order~$ e_1,e_2,e_3,\dots $ on edges has to be built adaptively depending on previous computations.
This list can be built edge by edge like this:
The sub-list~$ e_2,e_3,\dots $ is built under the assumption that~$ e_1 $ is not in the graph, the sub-list $ e_3,\dots $ is built under the assumption that both~$ e_1,e_2 $ are not in the graph, etc.
After \greedy adds the highest priority edge in the graph to the solution, a new round starts and a new priority list is built adaptively like above.
We have to translate~$ e_1,e_2,e_3,\dots $ into a priority order on adjacency lists.
Therefore edge~$ e_i=\edge{u,v} $ is replaced by a list containing an adjacency list~$ \langle u;v,w,x,\dots\rangle $ for each possible choice of~$ w,x,\dots $, making sure that if~$ e_i $ is in the graph it is picked before~$ e_{i+1} $.
If \greedy receives the adjacency list~$ \langle u;v,w,x,\dots\rangle $, node~$ u $ is matched with~$ v $.

The argument for the other algorithms in analogous.
In each round, a priority list is built recursively under the assumptions that high priority entries are not in the graph, and then translated to a list of adjacency lists.

\karpsipser works like \greedy but prefers an edge incident with a degree-1 node, if such an edge exists.
So edges incident with degree-1 nodes are moved to the front of the priority order, i.e.\ the priority order on adjacency lists starts with adjacency lists~$ \langle u;v,w,x,\dots\rangle $ for all possible choices of~$ u $ and~$ v,w,x,\dots $ where only~$ v $ is not already matched.

To implement~\mrg, a priority list on nodes~$ u_1,u_2,\dots $ has to be translated.
Node~$ u_i $ is replaced by a list of adjacency lists~$ \langle u_i;v,w,x,\dots\rangle $ for each possible choice of~$ v,w,x,\dots $
If~\mrg receives an adjacency list, say for node~$ u $, an arbitrary non-isolated neighbor of~$ u $ is to be matched with~$ u $.

To implement \mingreedy, a priority list on node degrees~$ 1,2,\dots $ has to be translated.
Degree~$ i $ is replaced by a list of adjacency lists~$ \langle u;v_1,\dots,v_i,w,x,\dots\rangle $ for each possible combination of~$ u $ and~$ v_1,\dots,v_i,w,x,\dots $ where only~$ v_1,\dots,v_i $ are not already matched.
As for~\mrg, for a received adjacency list, say for node~$ u $, an arbitrary non-isolated neighbor of~$ u $ is matched with~$ u $.

\shuffle does not compute priority lists in each round, but a node permutation~$ u_1,u_2,\dots,u_n $ is computed once at the start, using the number~$ n $ of nodes in the graph.
A node~$ u_i $ is replaced by a list of adjacency lists $ \langle u_i;v,w,x,\dots\rangle $ for each possible choice of~$ v,w,x,\dots $
This priority order is used in each round.
If \shuffle receives an adjacency list for a node~$ u_i $, then from the still non-isolated neighbors of~$ u_i $ the first one in~$ u_1,u_2,\dots,u_n $ is matched with~$ u_i $.

A vertex iterative algorithm \cite{DBLP:conf/focs/GoelT12} considers nodes one at a time, and probes each node~$ u $ for neighbors.
The probing for~$ u $ ends after the first successful probe, say for neighbor~$ v $.
Then nodes~$ u $ and~$ v $ are matched.
The probing for~$ u $ also ends after all possible neighbors have been tested without success.
Furthermore, the probing for~$ u $ ends if the algorithm decides to stop probing for further neighbors of~$ u $.
Once the probing for~$ u $ ends, node~$ u $ is never considered again and~$ u $ is never probed as the neighbor of any other node considered later.
In a round of of the algorithm, we denote by~$ (a,b),(a,c),\dots,(a,d),(e,f),(e,g),\dots,(e,h),\dots $ that node~$ a $ is to be probed for neighbors~$ b,c,\dots,d $, then~$ e $ is to be probed for neighbors~$ f,g,\dots,h $, etc.
A pair~$ (u,v) $ is translated to a list of adjacency lists $ \langle u;v,w,x,\dots\rangle $ for each possible choice of~$ w,x,\dots $.
If an adjacency list $ \langle u;v,w,x,\dots\rangle $ is received, then node~$ u $ is matched with~$ v $.
\qed
\end{proof}

To proof that no~\greedymatching-algorithm can guarantee approximation ratio better than~$ \desiredRatio+o(1) $, we first present in \cref{apvUpperBounds} a construction for~\greedymatching-algorithms without access to the number of nodes in the graph, and then adapt the construction in \cref{apvUpperBounds2} to the full class of~\greedymatching-algorithms.

\begin{lemma}
\label{apvUpperBounds}
Let $ A$ be an $\greedymatching $-algorithm without access to the number of nodes.
There is an input graph of degree at most \maxdeg for which $ A $ computes a matching of size at most \desiredRatio times optimal.
\end{lemma}

\begin{proof}
We describe the construction of a hard input instance \g for algorithm $ A $ as a game played between $ A $ and an adversary $ B $.
As $ A $ unveils \g only bit by bit, adversary $ B $ may actually construct \g on the fly, thereby reacting to the various moves of $ A $ such that \g has a much larger matching than the solution of $ A $.
Of course, all adjacency lists presented by $ B $ during the whole game have to be consistent with the final graph \g constructed by $ B $.

The game consists of the \emph{regular game}, which lasts for $ s=\maxdeg-3 $ steps, followed by the \emph{endgame}, which has two steps.

During the regular game, adversary $ B $ maintains the following invariant:
each node $ v $ that is not yet isolated has an adjacency list of one of the following types.
\begin{itemize}[itemindent=2.5em]
\item[Type 1:]
$ \langle v; v_{1}, \dots, v_{d} \rangle$ where $ v$ and $v_{1}, \dots, v_{d} $ are \emph{unknown}, i.e.\ they did not occur in a previously received adjacency list, and $ 3\leq d\leq \maxdeg $.
\item[Type 2:]
$ \langle v; v_{1},v_{2} \rangle$ where $ v$ and $v_{1}, v_{2} $ are unknown.
\item[Type 3:]
$ \langle v; v_{1},v_{2},v_{3} \rangle$ where $ v$ and $v_1,v_2 $ are unknown and $ v_3 $ is \emph{known}, i.e.\ node $ v_3 $ was received by $ A $ in a previous adjacency list.
\end{itemize}
Observe, that all nodes in \g have degree at least two.

Consider the very first step of $ A $.
Since all nodes are still unknown, all nodes have adjacency lists of type 1 or 2.
Hence the invariant holds.
Consider step $ i $ and assume that the invariant holds.
Adversary $ B $ presents the highest ranked adjacency list that is of type 1, 2 or 3.
Call that adjacency list $ a_i $.

\begin{figure}
\centering
\begin{minipage}[t]{.32\textwidth}
\begin{tikzpicture}
\tikzstyle{every node} = [circle, fill=white,draw=black,minimum size=13pt,inner sep=0pt];

\node[] (v) {$ v $};
\node[] (w) at ($(v) + (3,0)$) {$ v_1 $};
\node[] (x) at ($(v)+(0,-1.5)$){$ v_2 $};
\node[] (z) at ($(x)+(1,0)$){$ v_3 $};
\node[] (y) at ($(z)+(2,0)$) {$ v_d $};
\node[draw=none] at ($(y)+(-1,0)$) {$ \dots $};

\draw
(w) edge[bend right=15] (x)
(w) edge[bend right=15] (z)
(w) edge[] (y)
(w) edge[,bend right=15] (x)
(v) edge[mg] (w)
(v) edge[opt] (x)
(v) edge[,bend left=15] (z)
(w) edge[bend right=15] (z)
(w) edge[opt   ] (y)
(v) edge[,bend left=15] (y)
;

\draw[white] ($(v)+(-.4,+.4)$) rectangle ($(y)+(+.4,-.4)$);
\end{tikzpicture}
\caption{
A connected component of a hard instance
}
\label{unknownHigh_}
\end{minipage}
~
\begin{minipage}[t]{.65\textwidth}
\begin{tikzpicture}

\node[] (v') {$ m_1 $};
\node[very thick] (w') at ($(v') + (1,0)$) {$ r_1 $};
\node[left of=v'] (x') {$ l_1 $};
\node[,right of=w',,fill=lightgray] (g') {$ u_1 $};
\node[alabel] (vdots') at ($(x')+(0,-.65)$) {$\vdots$};
\node[alabel] (vdots') at ($(g')+(-0,-.65)$) {$\vdots$};
\node[] (v) at ($(v') + (0,-1.5)$) {$ m_k $};
\node[very thick] (w) at ($(v) + (1,0)$) {$ r_k $};
\node[left of=v] (x) {$ l_k $};
\node[right of=w,,fill=lightgray] (g) {$ u_k $};
\node[,fill=lightgray] at ($(g') + (2,0)$) (c) {$ a $};
\node[fill=lightgray]  at ($(c) + (0,-1.5)$) (c') {$ c $};
\node[fill=lightgray,] at ($(c)+(2,0)$) (c'') {$ b $};
\node[fill=lightgray] (d) at ($(c)+(0,-.75)$) {};
\node[fill=lightgray,right of=d] (d') {};
\node[fill=lightgray] (d'') at ($(c')+(2,0)$) {$ d $};

\draw[lightgray,thick,rounded corners=6pt] ($(c)+(-.4,.4)$) rectangle ($(d'')+(.4,-.4)$);
\node[draw=none,fill=none,text=gray] at ($(d')+(.5,0)$) {\rotatebox{90}{center}};

\draw
(g') edge[,bend right=20] (c')
(g) edge[,bend left=20] (c)
(g) edge[] (c')
(g') edge[] (c)
(w') edge[,bend right=35] (x')
(w') edge[   ,bend right=35] (x')
(v') edge[] (w')
(v') edge[] (x')
(w') edge[] (g')
(v') edge[mg] (w')
(v') edge[opt] (x')
(w') edge[] (g')
(w') edge[opt] (g')
(g') edge[] (c)
(w) edge[,bend right=35] (x)
(w) edge[   ,bend right=35] (x)
(v) edge[] (w)
(v) edge[] (x)
(w) edge[] (g)
(v) edge[mg] (w)
(v) edge[opt] (x)
(w) edge[] (g)
(w) edge[opt] (g)
(g) edge[] (c')
(g') edge[bend right=20,] (c')
(g) edge[bend left=20] (c)
(c) edge[] (c'')
(c) edge[] (c'')
(c) edge[] (d)
(c) edge[opt] (d)
(d) edge[] (c')
(d) edge[] (c')
(c') edge[] (d'')
(c') edge[] (d'')
(c'') edge[] (d'')
(c'') edge[opt] (d'')
(c) edge[,bend left=15] (d')
(c) edge[,bend left=15] (d')
(c') edge[,bend right=15] (d')
(c') edge[opt,bend right=15] (d')
;

\draw[dash pattern=on 2.5pt off 2.5pt] ($(w')+(0,-.4)$) .. controls  ($(w')+(.075,-.9)$) and ($(v)+(-.05,1.2)$)  .. ($(v)+(0,.6)$);

\end{tikzpicture}
\caption{
The core of a hard instance
(Gray nodes are unknown,
fat frontier nodes.
The dashed edge is an example for \edge{m_i=v,r_j=v_3} in case 3.)
}
\label{known_}
\end{minipage}
\end{figure}

\textbf{Case 1:}
$ a_i=\langle v; v_{1}, \dots, v_{d} \rangle $ is a type-1 adjacency list.
Since all nodes in $ a_i $ are unknown, we may w.l.o.g. assume that $ A $ matches $ v $ with $ v_1 $.
Adversary $ B $ constructs the connected component $ C $ depicted in \cref{unknownHigh_} which consists only of nodes of types 1 and 2.
All nodes of $ C $ are isolated in the next step, hence the invariant is maintained.
Observe that within $ C $ the maximum matching \mopt scores two edges (the double edges $ \edge{v,v_2},\edge{v_1,v_d} $ in \cref{unknownHigh_}) whereas the matching 
\m computed by $ A $ scores just one edge (the crossed edge $ \edge{v,v_1} $).

Since in Case 1 algorithm $ A $ requests only unknown nodes, adversary $ B $ is able to trick $ A $ into unveiling part of \g from which $ A $ cannot gather knowledge about the rest of \g.
Can $ A $ act smarter?
Assume that $ A $ has already received the adjacency lists of the \emph{middle nodes} $ m_1,$ $\dots,$ $m_k $ of the triangles $ \{l_j,m_j,r_j\} $ of known nodes connected by \emph{frontier nodes} $ r_j $ and unknown nodes $ u_j $ to the still unknown \emph{center} of \g, see \cref{known_}.
If $ A $ requests an unknown node with two unknown neighbors, then $ B $ easily tricks $ A $ by constructing a new triangle $ \{l_i,m_i,r_i\} $.

\textbf{Case 2:}
$ a_i=\langle v; v_{1}, v_{2} \rangle $ is a type-2 adjacency list.
Again, all nodes of $ a_i $ are unknown and we may assume that $ A $ matches $ v$ with $v_1 $.
Adversary $ B $ constructs a triangle $ \{l_i,m_i,r_i\} $ with $ l_i=v_2,m_i=v,r_i=v_1 $ and inserts the edge $ \edge{r_i,u_i} $, with a new unknown node $ u_i $, to connect the triangle to the unknown center.
Observe that before nodes $ m_i,r_i $ are matched, nodes $ m_i,l_i $ are of type 2 and $ r_i,u_i $ are of type 1.
After matching $m_i,r_i$, nodes $ l_i,m_i,r_i $ are isolated and $ u_i $ turns into a type-3 node.
Hence the invariant still holds.
Again, \mopt scores two edges, namely $ \edge{l_i,m_i},\edge{r_i,u_i} $, and $ M $ scores the edge $ \edge{m_i,r_i} $.

Now assume that $ A $ tries to explore the neighborhood of known nodes.
Observe that the only adjacency lists with a known node are of type 3 and have exactly one unknown node:
since the known nodes $ l_j,m_j,r_j $ are already isolated, an unknown node can only be explored in the neighborhood of frontier nodes.
Again, adversary $ B $ tricks $ A $ with a new triangle $ \{l_i,m_i,r_i\} $.

\textbf{Case 3:}
$ a_i=\langle v; v_{1}, v_{2},v_{3} \rangle $ is a type-3 adjacency list.
Since $ v_3 $ is known, $ v_3 $ occurred in a previously presented adjacency list.
Observe that in our construction so far, the only type-3 nodes are unknown neighbors of known frontier nodes.
So $ v $ is the neighbor of a frontier node $ r_j=v_3 $ with $ j<i $.

Is algorithm successful in exploring the unknown neighbor $ u_j $ of $ r_j $, i.e.\ does $ v=u_j $ hold?
Not necessarily, since $ B $ may on the fly construct further neighbors of $ r_j $.
Why?
Since $ r_j $ gets matched as soon as it becomes known, algorithm $ A $ never gets to see the adjacency list of $ r_j $ and consequently $ A $ can never tell if it already knows all neighbors of $ r_j $.
(Adversary $ B $ uses this trick here as well as in the end game.)

Since $ v_3=r_j $ is matched and $ v_1,v_2 $ are unknown, we may assume that $ A $ matches $ v $ with $ v_1 $.
Adversary $ B $ behaves exactly as in case 2 and constructs the triangle $ \{l_i=v_2,m_i=v,r_i=v_1\} $ and inserts the edge \edge{r_i,u_i} where $ u_i $ is a new unknown node.
To complete the devious trick, adversary $ B $ also inserts the edge $ \edge{m_i=v,r_j=v_3}$ (see e.g. the dashed edge in \cref{known_}).
Before $ m_i,r_i $ are matched, node $ l_i $ is of type 2, nodes $ r_i,u_i $ are of type 1 and $ m_i=v $ is of type 3.
After matching $m_i,r_i$, nodes $ l_i,m_i,r_i $ are isolated and $ u_i $ turns into a type-3 node.
$ u_j $ is still of type 3.
Hence the invariant still holds.
As in case 2, \mopt scores $ \edge{l_i,m_i},\edge{r_i,u_i} $ and \m scores $ \edge{m_i,r_i} $.

This concludes the regular game.
In the first step of the endgame adversary $ B $ makes algorithm $ A $ match $ a$ with $b $.
Hence in the next and last step algorithm $ A $ matches $ c $.
So algorithm $ A $ scores two edges in the center, whereas three edges are optimal.
As desired, we get
$$ |\m|=s+2=\maxdeg-1 ~~~~~~\mbox{ and }~~~~~~ |\mopt|=2s+3 =2\maxdeg-3\,.$$

Observe that our invariant still holds in the first step of the endgame.
Again, adversary $ B $ presents the highest ranked adjacency list of type 1, 2 or 3.
Observe that $ a$ and $c $ are the only type-1 nodes left, since the $ u_j $ have known neighbors and are of type 3 and all other nodes have degree two and are of type 2.
The degree of $ a$ and $c $ is $ \delta\leq3+s $, since both $ a,c $ have three center neighbors and each step of the regular game adds at most one neighbor to $ a$ respectively $c $.
Let $ a_{\maxdeg-2} $ be the adjacency list received in step $ s+1=\maxdeg-2 $.

\textbf{Case 4a:}
$ a_{\maxdeg-2}=\langle v; v_{1}, \dots, v_{\delta} \rangle $ is a type-1 adjacency list.
Since all nodes of $ a_{\maxdeg-2} $ are unknown we may assume that $ A $ matches $ v $ with $ v_1 $.
Adversary $ B $ chooses $ v =a$, $ v_1=b $ and $ v_2,\dots,v_\delta $ as the remaining neighbors of $ a $.

\textbf{Case 4b:}
$ a_{\maxdeg-2}=\langle v; v_{1}, v_{2} \rangle $ is a type-2 adjacency list.
Since all nodes of $ a_{\maxdeg-2} $ are unknown we may assume that $ A $ matches $ v $ with $ v_1 $.
Adversary $ B $ sets $ v =b$, $ v_1=a $ and $ v_2=d $.

\textbf{Case 4c:}
$ a_{\maxdeg-2}=\langle v; v_{1}, v_{2},v_{3} \rangle $ is a type-3 adjacency list.
As in case 3, the known node $ v_3 $ is some matched frontier node $ r_j,j<\maxdeg-2 $ and we may assume that $ A $ matches $v $ with $ v_1 $, since $ v_1,v_2 $ are unknown.
As in case 3, adversary $ B $ does \emph{not} present the adjacency list of the unknown node $ u_j $.
Instead, $ B $ makes $ b $ a neighbor of $ r_j $ by inserting \edge{v_3{=}r_j,b}---now $ b $ has three neighbors---and sets $ v{=}b,v_1{=}a$ and $v_2{=}d $.

Adversary $ B $ does not violate degree constraints.
Nodes introduced in case 1 have degree at most $ d\leq\maxdeg $.
All other degrees are at most three, but for $ a,c $ and frontier nodes $ r_j $.
As discussed, nodes $ a,c $ have degree at most $ \delta=	3+s\leq\maxdeg $.
Frontier node have degree at most $ 3+(s-1)+1=\maxdeg $, since in each but the first step of the regular game and in step $ s+1 $ at most one incident edge is added.
\qed
\end{proof}

\begin{theorem}
\label{apvUpperBounds2}
Let~$ A $ be an \greedymatching-algorithm.
There is a graph~$ G $ of degree at most \maxdeg for which $ A $ computes a matching of size at most $ \desiredRatio+\varepsilon $ times optimal for any $ \varepsilon>0 $.
\end{theorem}
\begin{proof}
We modify the adversary~$ B $ who constructs hard inputs in the proof of \cref{apvUpperBounds}.
First, the modified adversary $ B' $ announces the number $ t\maxdeg $ of nodes, where $ t$ is a large integer.
Using so many nodes~$ B' $ constructs~$ \Omega(t) $ connected components, each with~$ O(\maxdeg) $ nodes and approximation ratio no better than~\desiredRatio.
A negligible portion of the graph might be solved optimally by~$ A $.
In the proof we frequently refer to the adjacency list types and cases found in the proof of \cref{apvUpperBounds} (types 1, 2, and 3, and cases 1, 2, 3, and 4a-c).

Again, the game between $ A $ and $ B' $ is split up into the regular game and the endgame.
Like the adversary~$ B $ from the proof of \cref{apvUpperBounds}, in each round of the regular game the construction of $ B' $ keeps up the invariant that all non-isolated nodes in the graph have adjacency lists of type 1, 2 or, 3 and $ B' $ returns the highest priority adjacency list having one of these types.
However, depending on the requests of $ A $, not only one but several centers $ C_1,C_2,\dots $ might be constructed, each in its own connected component of~$ G $.
Each center $ C_i $ is defined as in the proof of \cref{apvUpperBounds}, with nodes $ a_i,b_i,c_i,d_i $ and two more nodes unique to~$ C_i $, see \cref{known_}.
Each $ C_i $ will get attached to it a maximum number of triangles, which are not connected to any other center.
Thereafter, the nodes of~$ C_i $ are supposed to be matched in the same order as in the proof of \cref{apvUpperBounds}, i.e.\ when no more triangles are attached, nodes $ a_i $ and $ b_i $ are matched to each other before $ c_i $ is matched.
Once $ c_i $ is matched, all edges in the connected component of $ C_i $ are removed.

Assume that $ B' $ has already created the centers $ C_1,\dots,C_l $.
Call $ C_i $ \emph{active} if $ a_i$ is not yet matched with $b_i $, and \emph{inactive} otherwise.
The construction will ensure that $ C_1,\dots,C_{l-1} $ are inactive;
$ C_l $ might still be active.
(We note here that after center $ C_l $ becomes inactive, there are nodes in the rest of the connected component $ K $ of $ C_l $ which do \emph{not} have adjacency lists of types 1, 2 or 3.
However, all these nodes are neighbors of $ c_l $ and adversary $ B' $ adds no more nodes to $ K $.
Thus $ A $ scores at most one more edge in $ K $.
Therefore, the additional adjacency list types do not have effect on the rest of the construction and we do not discuss these types explicitly.)

Assume that in the next round $ A $ receives a \textbf{type-2} adjacency list of a node with two unknown neighbors.
Assume that $ C_l $ is already inactive, then $ B' $ creates the next center $ C_{l+1} $ and connects a new type-2 triangle to $ a_{l+1},c_{l+1} $ as described in case 2.
If otherwise $ C_l $ is still active, let $ \delta $ be the number of neighbors of $ a_l,c_l $ constructed so far and recall that we demand $ \delta\leq\maxdeg $.
If $ \delta<\maxdeg $, then $ B' $ connects a new type-2 triangle to $ a_l,c_l $ as described in case 2.
If $ \delta=\maxdeg $, then $ B' $ makes $ A $ match $ a_l $ with $ b_l $ as described in case 4b, thereby inactivating $ C_l $.

Assume that in the next round $ A $ receives a \textbf{type-3} adjacency list.
By construction, the received node is unknown and among its three neighbors there is exactly one known node $ v_3 $, where $ v_3=r $ is a frontier node $ r $ in the connected component of the still active center $ C_l $. 
Let $ \delta $ be the number of neighbors of $ a_l,c_l $ constructed so far.
If $ \delta<\maxdeg $, then $ B' $ connects a new type-3 triangle to $ a_l,c_l $ as described in case 3.
If $ \delta=\maxdeg $, then $ B' $ makes $ A $ match $ a_l $ with $ b_l $ as described in case 4c, and inactivates $ C_l $.

Assume that in the next round $ A $ receives a \textbf{type-1} adjacency list.
If the degree of the received node is smaller than \maxdeg, then $ B' $ proceeds as in case 1 and creates a new type-1 connected component.
Now assume that the received node has degree \maxdeg.
If nodes $ a_l,c_l $ have degree less than \maxdeg, then again $ B' $ creates a new type-1 connected component.
If 
nodes $ a_l,c_l $ have degree $ \delta=\maxdeg $, then $ B' $ makes $ A $ match $ a_l $ with $ b_l $ as described in case 4a, and thereby inactivates $ C_l $.

Why is $ C_l $ inactivated before~$ B' $ constructs the next active center?
In type-2 and type-3 rounds an increasing number of triangles is connected to $ C_l $ until the degrees of $ a_l,c_l $ are \maxdeg.
(In intermediate type-1 rounds only type-1 connected components are constructed.)
Thereafter, $ C_l $ is inactivated in the first type-2 or type-3 round or the first type-1 round in which an adjacency list of a degree-\maxdeg node is received.
(In intermediate type-1 rounds with nodes of degree less than~\maxdeg only type-1 connected components are constructed.)

The endgame begins as soon as $ B' $ has constructed $ k\geq t\maxdeg-6\maxdeg$ nodes.
Let~$ \nu = t\maxdeg-k $ be the number of nodes still to be constructed.
Observe that $ 6\maxdeg\geq\nu\geq2\maxdeg $ holds, since in each round no more than $4\maxdeg$ additional nodes are introduced (e.g. if~$ \maxdeg=3 $ holds and a new triangle is connected to a new center).
Since $ B' $ has committed to a number of exactly $ t\maxdeg $ nodes, in the first round of the endgame $ B' $ utilizes all remaining $ \nu $ nodes to create additional connected components $ \Gamma_1,\dots,\Gamma_c $, each being a complete bipartite graph with $ 2 $ nodes on the right side and between~2 and~\maxdeg (both~2 and~\maxdeg included) nodes on the left. 
The left sides are as large as possible such that $ c=O(1) $ is constant and all nodes in the~$ \Gamma_i $ have degree at least two.
All nodes in  $ \Gamma_1,\dots,\Gamma_c $ are still unknown, in particular all nodes have adjacency lists only of types 1 or 2.
Since right sides have two nodes, each $ \Gamma_i $ has at most two edges in a maximum matching, making an additional constant number $ 2c $ of optimal edges in total.
We assume that $ A $ performs optimally in all $ \Gamma_i $, thereby scoring~$ 2c $ edges.

Also in the endgame, algorithm $ A $ matches still unmatched nodes in the already inactive centers $ C_1,\dots,C_{l-1} $.
We assume that $ A $ performs optimally also in the connected component of the last center $ C_l $.

What is the approximation ratio of $ A $?
Observe that each type-1 component has~$ O(\maxdeg) $ nodes.
The same is true for each connected component with a (constant size) center~$ C_i $, since~$ O(\maxdeg) $ triangles (of constant size) are attached to~$ C_i $.
So as claimed,~$\Omega(t)$ connected components are constructed.
Recall that $ A $ scores only one out of two edges in each type-1 connected component.
On the other hand, in the connected component of an inactive center~$ A $ achieves approximation ratio exactly~\desiredRatio.
In all other components,~algorithm $ A $ even performs optimally.
Therefore, to bound the performance of $ A $ we may assume that no type-1 components are constructed.
Since only $ C_l $ might be active at the end of the regular game, there are at least $ l-1=\Omega(t) $ inactive centers, hence the approximation ratio of $ A $ is at most
\begin{align*}
\frac{(l-1)\cdot~~(\maxdeg-1)+(2\maxdeg-3)+2c}{(l-1)\cdot(2\maxdeg-3)+(2\maxdeg-3)+2c}\,,
\end{align*}
since $ A $ performs optimally in $ C_l $ and in $ \Gamma_1,\dots,\Gamma_c $.
Letting $ t\to\infty $ we get~$ l\to\infty $ and this ratio is dominated by $ \frac{(l-1)\cdot(\maxdeg-1)}{(l-1)\cdot(2\maxdeg-3)} $.
The statement follows.
\qed
\end{proof}

\section{Conclusion}
We have analyzed the worst case approximation ratio of the well-known \mingreedy algorithm on graphs of bounded degree.
Our performance guarantees of \twoOverThree for graphs of degree at most three and of \desiredRatio for \maxdeg-regular graphs are tight.
In particular, \mingreedy is optimal in the large class of \greedymatching-algorithms, which contains many prominent greedy matching algorithms.
We also proved a performance guarantee of \weakerRatio for graphs of degree at most \maxdeg, and we conjecture that also in this case \mingreedy is optimal among \greedymatching-algorithms and achieves a worst case approximation ratio of at least \desiredRatio.

Our worst case performance guarantees are stronger than the best known worst case bounds on the expected approximation ratio for the well-known greedy matching algorithms \greedy, \mrg and \shuffle, if degrees are small.

\inlineheading{Open Questions.}
Is \mingreedy optimal among \greedymatching-algorithms on graphs of degree at most \maxdeg?

What bounds for \mingreedy can be shown for more restricted graph classes, e.g. bipartite graphs?

Recall that the expected approximation ratio of the randomized \mingreedy algorithm is $ \oneOverTwo+o(1) $ w.h.p. on graphs of arbitrarily large degree.
Does randomized \mingreedy have an expected approximation ratio strictly better than \desiredRatio if degrees are bounded by \maxdeg?
~\\

\noindent
\inlineheading{Acknowledgements.}
I thank Georg Schnitger for many helpful discussions.


\begin{thebibliography}{CCWZ14}

\bibitem[ADFS95]{RSA:RSA3240060107}
Jonathan Aronson, Martin Dyer, Alan Frieze, and Stephen Suen.
\newblock Randomized greedy matching. ii.
\newblock {\em Random Structures \& Algorithms}, 6(1):55--73, 1995.

\bibitem[BNR02]{DBLP:conf/soda/BorodinNR02}
Allan Borodin, Morten~N. Nielsen, and Charles Rackoff.
\newblock (incremental) priority algorithms.
\newblock In {\em Proc. 13th SODA}, pages 752--761, 2002.

\bibitem[CCWZ14]{DBLP:conf/soda/ChanCWZ14}
T.-H.~Hubert Chan, Fei Chen, Xiaowei Wu, and Zhichao Zhao.
\newblock Ranking on arbitrary graphs: Rematch via continuous lp with monotone
  and boundary condition constraints.
\newblock In Chandra Chekuri, editor, {\em SODA}, pages 1112--1122. SIAM, 2014.

\bibitem[DF91]{DBLP:journals/rsa/DyerF91}
Martin~E. Dyer and Alan~M. Frieze.
\newblock Randomized greedy matching.
\newblock {\em Random Structures \& Algorithms}, 2(1):29--46, 1991.

\bibitem[DI04]{DBLP:conf/soda/DavisI04}
Sashka Davis and Russell Impagliazzo.
\newblock Models of greedy algorithms for graph problems.
\newblock In {\em Proc. 15th SODA}, pages 381--390, 2004.

\bibitem[FRS93]{Frieze:1993:ASG:313559.313800}
Alan Frieze, A.~J. Radcliffe, and Stephen Suen.
\newblock Analysis of a simple greedy matching algorithm on random cubic
  graphs.
\newblock In {\em Proc. 4th SODA}, pages 341--351, 1993.

\bibitem[GT12]{DBLP:conf/focs/GoelT12}
Gagan Goel and Pushkar Tripathi.
\newblock Matching with our eyes closed.
\newblock In {\em Proc. 53rd FOCS}, pages 718--727, 2012.

\bibitem[KS81]{DBLP:conf/focs/KarpS81}
Richard~M. Karp and Michael Sipser.
\newblock Maximum matchings in sparse random graphs.
\newblock In {\em Proc. 22nd FOCS}, pages 364--375, 1981.

\bibitem[KVV90]{DBLP:conf/stoc/KarpVV90}
Richard~M. Karp, Umesh~V. Vazirani, and Vijay~V. Vazirani.
\newblock An optimal algorithm for on-line bipartite matching.
\newblock In {\em Proc. 22nd STOC}, pages 352--358, 1990.

\bibitem[Mag97]{Magun97greedymatching}
Jakob Magun.
\newblock Greedy matching algorithms, an experimental study.
\newblock In {\em Proceedings of the 1st Workshop on Algorithm Engineering},
  volume~6, pages 22--31, 1997.

\bibitem[MP97]{DBLP:journals/rsa/MillerP97}
Zevi Miller and Dan Pritikin.
\newblock On randomized greedy matchings.
\newblock {\em Random Struct. Algorithms}, 10(3):353--383, 1997.

\bibitem[MS04]{DBLP:conf/focs/MuchaS04}
Marcin Mucha and Piotr Sankowski.
\newblock Maximum matchings via gaussian elimination.
\newblock In {\em FOCS}, pages 248--255. IEEE Computer Society, 2004.

\bibitem[MV80]{DBLP:conf/focs/MicaliV80}
Silvio Micali and Vijay~V. Vazirani.
\newblock An o(sqrt(|v|) |e|) algorithm for finding maximum matching in general
  graphs.
\newblock In {\em FOCS}, pages 17--27. IEEE Computer Society, 1980.

\bibitem[Pol12]{matthiasDiss}
Matthias Poloczek.
\newblock {\em Greedy algorithms for max sat and maximum matching : their power
  and limitations}.
\newblock PhD thesis, Institut f{\"u}r Informatik, Goethe-Universit{\"a}t
  Frankfurt am Main, 2012.
\newblock \url{http://d-nb.info/1036608425}.

\bibitem[Tin84]{tinhofer}
G.~Tinhofer.
\newblock A probabilistic analysis of some greedy cardinality matching
  algorithms.
\newblock {\em Annals of Operations Research}, 1(3):239--254, 1984.

\bibitem[Vaz12]{DBLP:journals/corr/abs-1210-4594}
Vijay~V. Vazirani.
\newblock An improved definition of blossoms and a simpler proof of the mv
  matching algorithm.
\newblock {\em CoRR}, abs/1210.4594, 2012.

\end{thebibliography}
\end{document}